\documentclass[12pt, a4paper]{article}
\usepackage{jheparxiv}
\usepackage[utf8]{inputenc}
\usepackage{amsmath}
\usepackage{amsfonts}
\usepackage{amssymb}
\usepackage{latexsym}
\usepackage{mathrsfs}
\usepackage{braket}		
\usepackage{setspace}
\usepackage{graphicx}
\usepackage{color}
\usepackage{amsthm} 
\usepackage{xcolor}
\usepackage{slashed}
\usepackage{twistor}
\usepackage{skak}
\usepackage[all]{xy}
\usepackage{tikz-cd}
\usepackage{mathtools}
\usepackage{tikz}
\newcommand*\circled[1]{\tikz[baseline=(char.base)]{
            \node[shape=circle,draw,inner sep=2pt] (char) {#1};}}

\newcommand{\bep}{\begin{picture}}
\newcommand{\eep}{\end{picture}}

\newcounter{YoungHeight}\newcounter{YoungWidth}

\newcounter{Mul1}\newcounter{Mul2}\newcounter{Mul3}\newcounter{Mul4}
\newcounter{A0}\newcounter{A1}\newcounter{A2}
\newcounter{B3}
\newcounter{C3}\newcounter{C4}
\newcounter{D1}\newcounter{D2}\newcounter{D3}
\newcounter{T0}\newcounter{T1}

\newlength{\txtHShift}

\newlength{\txtWidth}

\newcommand{\HalfLength}[2]{\setcounter{Mul1}{#1}\setcounter{Mul2}{#1}\addtocounter{Mul1}{\value{Mul2}}\addtocounter{Mul1}{\value{Mul2}}%
\addtocounter{Mul1}{\value{Mul2}}\addtocounter{Mul1}{\value{Mul2}}\setcounter{#2}{\value{Mul1}}}

\newcommand{\Length}[1]{#10}

\newcommand{\shiftedText}[2]{{\hspace{#1}#2}}
\newcommand{\calcHShift}[1]{\settowidth{\txtWidth}{#1}\setlength{\txtHShift}{-0.5\txtWidth}}

\newcommand{\TextCenter}[3]{{\HalfLength{#2}{T0}%
\HalfLength{#3}{T1}\addtocounter{T1}{-3}\calcHShift{#1}%
\put(\value{T0},\value{T1}){\shiftedText{\txtHShift}{#1}}}}

\newcommand{\RectT}[3]{\bep(\Length{#1},\Length{#2})\put(0,0){\line(1,0){\Length{#1}}}\put(0,0){\line(0,1){\Length{#2}}}%
\put(\Length{#1},\Length{#2}){\line(-1,0){\Length{#1}}}\put(\Length{#1},\Length{#2}){\line(0,-1){\Length{#2}}}#3{#1}{#2}\eep}

\newcommand{\RectBRowUp}[4]{{\bep(\Length{#1},20)\put(0,0){\RectT{#2}{1}{\TextCenter{#4}}}%
\put(0,10){\RectT{#1}{1}{\TextCenter{#3}}}\eep}}

\newcommand{\tm}{\mathtt{m}}
\newcommand{\tn}{\mathtt{n}}
\newcommand{\tv}{\mathtt{v}}
\newcommand{\tto}{\mathtt{o}}

\newcommand{\redone}{\boldsymbol{\textcolor{red!70!black!100}{1}}}

\newcommand{\tI}{\mathtt{I}}
\newcommand{\tV}{\mathtt{V}}
\newcommand{\tq}{\mathtt{q}}

\newcommand{\kbold}{\boldsymbol{k}}
\newcommand{\mso}{\mathfrak{so}}
\newcommand{\msp}{\mathfrak{sp}}

\newcommand{\Tr}{\text{Tr}}

\newcommand{\ta}{\mathtt{a}}
\newcommand{\tb}{\mathtt{b}}
\newcommand{\tc}{\mathtt{c}}
\newcommand{\td}{\mathtt{d}}
\newcommand{\te}{\mathtt{e}}
\newcommand{\tf}{\mathtt{f}}

\newcommand{\BG}{\mathtt{BG}}
\newcommand{\boldJ}{\boldsymbol J}

\newcommand{\tp}{\mathsf{p}}

\usepackage{graphicx}

\newcommand{\nn}{\nonumber}

 \def\one{\mbox{1 \kern-.59em {\rm l}}}

\newcommand{\sa}{\mathsf{a}}

\newcommand{\ttb}{\underline{\mathtt{t}}}

\newcommand{\tA}{\mathtt{A}}
\newcommand{\tB}{\mathtt{B}}
\newcommand{\tC}{\mathtt{C}}
\newcommand{\tD}{\mathtt{D}}

\renewcommand{\d}{\mathrm{d}}

\newcommand{\boldI}{\boldsymbol I}
\newcommand{\Ibold}{\boldsymbol{I}}

\newcommand{\ib}{\mathtt{I}}
\newcommand{\jb}{\mathtt{J}}

\newcommand{\cY}{\hat{\mathcal{Y}}}
\newcommand{\eps}{\epsilon}

\newcommand{\msf}[1]{\mathsf{#1}}

\newcommand{\sA}{\tilde{\mathsf{A}}}

\newcommand{\msu}{\mathfrak{su}}
\newcommand{\PP}{\mathbb{P}}
\newcommand{\PPb}{\overline{\mathbb{P}}}

\newcommand{\hsikkt}{$\hs$-IKKT }   

\title{Spinorial description for Lorentzian \hsikkt}

\author[]{Harold C. Steinacker}
\affiliation[]{Department of Physics, University of Vienna, \\
Boltzmanngasse 5, A-1090 Vienna, Austria}
\emailAdd{harold.steinacker@univie.ac.at}

\author[]{\& Tung Tran}

\emailAdd{tung.tran@univie.ac.at}

\abstract{We introduce a novel spinorial description for the higher-spin gauge theory induced by the IKKT matrix model on an FLRW spacetime with Lorentzian signature, called Lorentzian \hsikkt theory. The new description is based on Weyl spinors transforming under the space-like isometry subgroup $SL(2,\C)$ of the structure group $SO(2,4)\simeq SU(2,2)$.
It allows us to exploit the full power of the spinor formalism in Lorentzian signature, in contrast to a previous formalism based on the compact subgroup $SU(2)_L\times SU(2)_R$ of $SU(2,2)$. Some cubic vertices of the Yang-Mills sector and the corresponding scattering amplitudes are computed. 
We observe that the $n$-point (for $n\geq 4$) tree-level amplitudes are typically non-trivial on-shell, but exponentially suppressed in the late-time regime. While Lorentz invariance of the higher-spin amplitudes is not manifest, it is expected to be restored by higher-spin gauge invariance.}

\begin{document}

 \begin{flushright}
  UWThPh-2023-30 
 \end{flushright}
 \vspace{-10mm}

\maketitle
\section{Introduction}


In recent years, considerable evidence has been gathered suggesting that the no-go theorems \cite{Weinberg:1964ew,Coleman:1967ad} that were thought to rule out the existence of local higher-spin 
($\hs$)  theories in 4-dimensional flat space can be overcome by relaxing some of the assumptions in quantum field theory, such as parity invariance \cite{Tran:2022amg} or local Lorentz invariance \cite{Steinacker:2023cuf}. 

On the one hand, when parity invariance is relaxed, one can construct (quasi-) chiral massless higher-spin theories \cite{Metsaev:1991mt,Metsaev:1991nb,Ponomarev:2016lrm,Ponomarev:2017nrr,Metsaev:2018xip,Metsaev:2019aig,Tsulaia:2022csz,Sharapov:2022faa,Sharapov:2022awp,Herfray:2022prf,Adamo:2022lah} that have simple $S$-matrices \cite{Skvortsov:2018jea,Skvortsov:2020wtf,Adamo:2022lah}. On the other hand, if mild Lorentz violation is allowed, one can construct a supersymmetric higher-spin gauge theory induced by the IKKT matrix model \cite{Ishibashi:1996xs} called \hsikkt \cite{Sperling:2018xrm,Steinacker:2019awe,Steinacker:2023zrb}. These classes of higher-spin theories do not suffer from unitarity issues as in the case of conformal higher-spin gravity~\cite{Tseytlin:2002gz,Segal:2002gd,Bekaert:2010ky,Basile:2022nou}.\footnote{There are also $3d$ higher-spin gravities \cite{Blencowe:1988gj,Bergshoeff:1989ns,Pope:1989vj,Fradkin:1989xt,Campoleoni:2010zq,Henneaux:2010xg,Gaberdiel:2010pz,Gaberdiel:2012uj,Gaberdiel:2014cha,Grigoriev:2019xmp,Grigoriev:2020lzu}. However, they do not have propagating degrees of freedom, and they are topological.} One crucial distinction between (quasi-) chiral higher-spin gravities (HSGRAs) and \hsikkt is that \hsikkt can be defined to be unitary or \emph{Lorentzian real}
, while (quasi-) chiral HSGRAs are well-defined only on self-dual spacetime with Euclidean or split signatures. 

In \cite{Steinacker:2023cuf}, we showed that \hsikkt theory on an $SO(1,3)$-invariant cosmological spacetime $\cM^{1,3}$ leads to time-dependent couplings which suppress the (mildly Lorentz-violating) cubic $\hs$ vertices by a factor of $e^{-\frac{3}{2}\tau}$, and the quartic vertices by a factor of $e^{-3\tau}$. Intriguingly, there is no explicit  Lorentz violation in the classical, lowest-spin sector for the cubic vertices. 
Moreover, for physical fields, Lorentz invariance can be recovered in \hsikkt theory from volume-preserving gauge transformations. 
However, this theory has some non-standard features. In particular, a gauge field with internal spin-$s$ 
carries the same dof. as a massive higher-spin field. 
 Since there is no physical mass parameter in the IKKT model, these degrees of freedom are referred to as  ``would-be-massive'' ($\hs$) modes \cite{Steinacker:2019awe}.\footnote{The would-be-massive higher-spin modes in \hsikkt theory help to mitigate the effect of Lorentz violation caused by the non-commutativity of matrices in certain covariant quantum (fuzzy) spaces such as $S_N^4$ or $H_N^4$.} Furthermore, \hsikkt leads to a spectrum with \emph{truncated} towers of higher-spin modes, which arise naturally from the underlying doubleton representations of $SU(2,2)$. This is in contrast to conventional HSGRAs in $d\geq 4$, where infinitely many higher-spin fields are always needed for consistency (cf. \cite{Fradkin:1986ka,Eastwood:2002su}). 
Due to these remarkable features,  Lorentzian \hsikkt  appears to be a physically reasonable and perhaps near-realistic higher-spin theory.

The \hsikkt theory is most easily formulated using the vectorial description, see e.g. \cite{Sperling:2018xrm,Steinacker:2023zrb}. However, the diagonalization of higher-spin modes in the gauge sector is quite involved \cite{Sperling:2018xrm,Steinacker:2019awe}, and the computation of amplitudes involving gauge fields is expected to be complicated due to the mixing between higher-spin modes. To alleviate these difficulties,  a spinorial formulation of the model is desirable. Such a formulation of \hsikkt was proposed previously in \cite{Steinacker:2023zrb}, in terms of spinors transforming under the compact subgroup $SU(2)_L\times SU(2)_R\subset SU(2,2)$. While this works quite well for the Euclidean case, that compact subgroup does not allow a global spinorial description for the Lorentzian \hsikkt on $\cM^{1,3}$, seriously restricting the power of the spinor formalism. The basic issue is that the spinors in \cite{Steinacker:2023zrb} were \emph{not} Weyl spinors. 
 
In the present work, we overcome this issue by working with the local non-compact space-like subgroup $SL(2,\C)\simeq SO(1,3)$ of the FLRW background and using an appropriate chiral basis for $\gamma$-matrices, in contrast to \cite{Sperling:2018xrm}. Our main results include:
\begin{itemize}
    \item[-] A suitable definition of Weyl spinors allows us to obtain a global spinorial description for Lorentzian \hsikkt theory on $\cM^{1,3}$. 
    \item[-] A spinorial organization of the $\hs$ modes and a simple spinorial description for the $\Box$ operator. This allows us to diagonalize all higher-spin modes in the Yang-Mills sector without much effort.
    \item[-] Explicit spacetime expressions for some cubic and quartic vertices of the Yang-Mills sector in Lorentzian \hsikkt. We also compute some simple scattering amplitudes, and observe that $n$-point scattering amplitudes are suppressed by a factor of $e^{-\frac{3}{2}(n-2)\tau}$, where $\tau$ is a cosmological time-like parameter. 
\end{itemize}

The paper is organized as follows. Section \ref{sec:2} provides a review of the IKKT matrix model and \hsikkt theory on $\cM^{1,3}$. Section \ref{sec:3} gives a twistor realization of $\cM^{1,3}$ and a definition of Weyl spinors transforming under the space-like isometry group $SO(1,3)\simeq SL(2,\C)$. We also discuss the internal higher-spin structures on $\cM^{1,3}$. Section \ref{sec:SO(1,3)spinors-IKKT} provides a spinorial formulation for the frame on $\cM^{1,3}$ and also for the $\hs$ modes. We then derive the spacetime action of the Yang-Mills sector in \hsikkt theory on $\cM^{1,3}$ in the local 4-dimensional physical regime. Section \ref{sec:amplitudes} computes some 3-point amplitudes of the Yang-Mills sector. We observe that these on-shell amplitudes are non-trivial, but strongly suppressed in the late-time regime. Upon projecting these three-point amplitudes to the massless sector, we observe that the Yang-Mills sector of the Lorentzian theory exhibits some chiral features. We conclude with some discussions in Section \ref{sec:discussion}. Three appendices are also included:
\begin{itemize}
    \item Appendix \ref{app:C} provides additional material for Section \ref{sec:2} on the fuzzy 4-hyperboloid $H^4_N$ and the doubleton minimal representations associated to $SU(2,2)$.
    \item Appendix \ref{app:A} provides technical details of averaging over fiber coordinates to obtain spacetime action of Lorentzian \hsikkt theory.
    \item Appendix \ref{app:B} provides further details for the computation of some 3-point scattering amplitudes in the Yang-Mills sector of the model.
\end{itemize}

\section{Review of the vectorial \texorpdfstring{\hsikkt}{hsikkt}}\label{sec:2}
This section provides a short review of the IKKT matrix model on $\cM^{1,3}$ in the vectorial formulation \cite{Sperling:2019xar}.

\paragraph{Notation.} Since there are many types of indices and symbols used in this paper, we summarize them in table \ref{table1}. We use the strength-one symmetrization convention, e.g. $A_aB_a=\frac{1}{2}(A_{a_1}B_{a_2}+A_{a_2}B_{a_1})$, and write  $T_{a(s)}=T_{a_1\ldots a_s}$ to describe a totally symmetric rank-$s$ tensor. 

\renewcommand{\arraystretch}{1.1}
\begin{table}[ht!]
    \centering
    \begin{tabular}{|c|c||c|} \hline
       Indices  & Group  & Symbols\\ \hline\hline
       $\Ibold,\boldJ=0,1,\ldots,9$ & $SO(1,9)$ & $y^{\Ibold}$\\ \hline
       $\tA,\tB=0,1,\ldots 5$  & $SO(2,4)$ & $m^{\tA\tB},\theta^{\tA\tB}$ \\
      $a,b=0,1,\ldots 4$ & $SO(1,4)$ & $m^{ab},y^a,\theta^{ab}$ \\
      $\hat a,\hat b=0,1,2,3,5$ & $SO(2,3)$ & $m^{\hat a\hat b},t^{\hat a}$ \\
      $\mu,\nu=0,1,2,3$ & $SO(1,3)$ & $\eta^{\mu\nu},\theta^{\mu\nu},m^{\mu\nu},t^{\mu},y^{\mu}$ \\ \hline
      $A,B=1,2,3,4$ & $SU(2,2)$ & $Z^A,\bar Z_A,C^{AB}$ \\
      $\alpha,\beta=0,1$ and $\dot\alpha\,\dot\beta=\dot 0,\dot 1$ & $SL(2,\C)$ & $\lambda^{\alpha},\mu^{\dot\alpha},\bar\mu_{\alpha},\bar\lambda_{\dot\alpha},\eps^{\alpha\beta},\eps^{\dot\alpha\dot\beta}$\\\hline
       $\mathtt{I},\mathtt{J}=4,5,6,7,8,9$ & $SO(6)$ & $\phi^{\mathtt{I}}$ \\\hline
    \end{tabular}
    \caption{Summary of notation. The meaning of each symbol will be explained in the text.}
    \label{table1}
\end{table}

\paragraph{IKKT matrix model at large $N$.} The IKKT matrix model with the action \cite{Ishibashi:1996xs}
\begin{align}\label{IKKTSO(1,9)}
    S=\Tr\Big([Y^{\boldsymbol{I}},Y^{\boldsymbol{J}}][Y_{\boldsymbol{I}},Y_{\boldsymbol{J}}]+\Psi_{\cA}(\Gamma^{\boldsymbol{I}})^{\cA\cB}[Y_{\boldsymbol{I}},\Psi_{\cB}]+2m^2Y^{\Ibold}Y_{\Ibold}\Big)\,,\qquad {\boldsymbol{I}}=0,1,\ldots,9\,
\end{align}
is invariant under $\delta Y^{\boldsymbol{I}}=U^{-1}Y^{\boldsymbol{I}}U$ and $\delta \Psi_{\cB}=U^{-1}\Psi_{\cB}U$ for any unitary matrix $U$, where $Y^{\boldsymbol{I}}$ are $N\times N$ Hermitian matrices and $\Psi_{\cB}$ are $\mso(1,9)$-valued Majorana-Weyl spinors. Here, we have added a small mass term by hand to set a scale for the model, and to stabilize the background brane $\cM\xhookrightarrow{} \R^{1,9}$ as a classical solution. Then the commutators of 
$Y^{\Ibold}$ 
\begin{align}\label{Poisson1}
    [Y^{\boldI},Y^{\boldJ}]:=\im\,\theta^{\boldI \boldJ}\,,
\end{align}
encode a Poisson structure
$\theta^{\Ibold\boldJ}$ on the brane $\cM\xhookrightarrow{}\R^{1,9}$. 

We denote $\cH$ as the Hilbert space on which these matrices act. Then, there exist certain localized quasi-coherent states $|\zeta\rangle\in \cH$ such that the expectation values $y^{\Ibold}:=\langle \zeta |Y^{\Ibold}|\zeta \rangle$ provide coordinate functions on $\cM$ in the semi-classical limit where $N$ is large. In this regime, the commutators can be replaced by Poisson brackets 
\begin{align}
 [Y^{\Ibold},Y^{\boldJ}]\quad \mapsto   \quad \{y^{\boldI},y^{\boldJ}\}:=\theta^{\boldI\boldJ}\,.
\end{align}
Note that the pair $(C^{\infty}(\cM,\R),\{\,,\})$ defines a Poisson algebra for the IKKT matrix model.

\paragraph{$SO(1,3)$-invariant spacetime $\cM^{1,3}$.} 
Following \cite{Govil:2013uta}, it is useful to organize the $\mso(2,4)$ generators $m_{\tA\tB}$ cf. \eqref{so(2,4)-Alg} where $\tA,\tB=0,1,\ldots,5$ as follows
\begin{subequations}
    \begin{align}
    m_{45}:&=D\qquad \qquad \ \  (\text{the dilation})\,,\\
    y_a:&=\ell_pm_{a5}\,,\qquad a=0,1,2,3,4\,,\\
    t_{\hat a}:&=\frac{1}{R}m_{\hat a 4}\,,\qquad \hat a=0,1,2,3,5\,.
\end{align}
\end{subequations}
Here $y^a$ describes a 4-dimensional hyperboloid $H^4$ with radius $R$, and $\ell_p$ is a natural length scale. The generators $y^\mu$ for $\mu=0,...,3$ will be interpreted as Cartesian coordinate functions on  $\cM^{1,3}$, which will be understood as cosmological space-time.
The ``momentum'' generators $t_{\hat a}$ transform as vectors under the $SO(2,3)$ subgroup of $SO(2,4)$ with generators $m^{\hat a\hat b}$, and $\ell_pR \,t_5=-y_4$. 
They satisfy the following relations:\footnote{See e.g. \cite{Steinacker:2023zrb} or Appendix \ref{app:C} for a short review.}
\begin{subequations}
\label{M31-rlations}
\begin{align}
    \{t^{\hat a},t^{\hat b}\}&=\frac{1}{R^2}m^{\hat a \hat b}\,,\qquad \qquad \qquad \qquad \quad \quad \qquad \qquad \, \,  \hat a,\hat b=0,1,2,3,5\,,\label{eq:Phatcommutator}\\
    \eta_{\hat{a}\hat{b}}t^{\hat{a}}t^{\hat{b}}&=-t_{0}^2+t_mt^m-t_{5}^2=\frac{1}{\ell_p^2}\,,\qquad \qquad \qquad \qquad \ \, m = 1,2,3\,,\label{Psphere}\\
    y_{\hat a}t^{\hat a}&=0=y_{\mu}t^{\mu}\,,\qquad \qquad \qquad \qquad \quad \quad \qquad \qquad \ \,  \mu=0,1,2,3\,, \label{eq:orthogonalofPY}\\
    \eta_{\mu\nu}t^{\mu}t^{\nu}&=\frac{1}{\ell_p^2}+\frac{y_4^2}{\ell_p^2R^2}=+\ell_p^{-2}\cosh^2(\tau)\,, \label{S2sphereM13}\\
   y^\mu y_\mu &= - R^2 \cosh^2(\tau) 
   \label{yy-square} \\ 
     \{t^{\mu},y^{\nu}\}&=+\frac{\eta^{\mu\nu}}{R}y^4=\eta^{\mu\nu}\sinh(\tau)\,,\qquad \qquad \qquad \qquad \, y_4:=R\sinh (\tau)\,,\label{eq:PoissonM13}\\
    \{t^{\mu},y^4\}&=-\frac{y^{\mu}}{R}\,,\\
    m^{\mu\nu}&= R^2 \{t^\mu,t^\nu\}   =\frac{1}{\cosh^2(\tau)}\Big(\sinh(\tau)(y^{\mu}t^{\nu}-y^{\nu}t^{\mu})+\eps^{\mu\nu\sigma\rho}y_{\sigma}t_{\rho}\Big)\label{mgenerator}\,.
\end{align}
\end{subequations}
where $\eta_{\hat{a}\hat b}=\diag(-,+,+,+,-)$ and $\eta_{\mu\nu}=\diag(-,+,+,+)$. 
The parameter $\tau$ can be interpreted as a time parameter in a $k=-1$ cosmological FLRW spacetime, with the $SO(1,3)$-invariant time-like vector field
\begin{align}\label{timelikeT}
    \cY:=y^{\mu}\p_{\mu}\,
\end{align}
and a cosmic scale parameter $a(t)$; for details of the geometry, we refer the reader to \cite{Sperling:2019xar}.

\paragraph{Background spacetime $\cM^{1,3}$ and Poisson action.} 

The above $t^\mu$ generators can be used to define an $SO(1,3)$-invariant cosmological spacetime within the matrix model.
This is realized by the following background solution of the model \eqref{IKKTSO(1,9)}
\begin{align}\label{BG}
    y^{\Ibold}=\binom{\ttb^{\mu}}{{0}}\,,\qquad \mu=0,1,2,3\,,
\end{align}
where the IR mass $m^2 = \frac 3{R^2}$
sets the scale for the model \cite{Steinacker:2023myp}.
We will always consider the semi-classical limit or regime, where all matrices are treated as functions and commutators as Poisson brackets.
The action \eqref{IKKTSO(1,9)} then defines a gauge theory on this background 
by adding fluctuations:
\begin{align}\label{BG-fields}
    y^{\Ibold}=\binom{\ttb^{\mu}}{0}+\binom{\sa^{\mu}}{\phi^{\ib}}\,,\qquad \mu=0,1,2,3\,,\qquad \ib=4,\ldots,9\,.
\end{align}
Here
\begin{align}\label{eq:higher-spin-in-t}
    \sa_{\mu}(y|\ttb)=\sum_s\cA_{\nu(s)|\mu}(y)\ttb^{\nu(s)}\,
\end{align}
can be identified as $\hs$-valued gauge fields with the infinitesimal gauge transformation arising from $y^{\Ibold} \to y^{\Ibold} + \{\xi,y^{\Ibold}\}$.
The action of Lorentzian \hsikkt takes the explicit form  \cite{Steinacker:2023zrb}:
\begin{align}\label{FLRWaction}
    S=\int &\mho\,\Big(\frac 12 \{\ttb^{\mu},\sa^{\nu}\}\{\ttb_{\mu},\sa_{\nu}\}+\{\ttb_{\mu},\ttb_{\nu}\}\{\sa^{\mu},\sa^{\nu}\}+\frac{1}{2}\{\ttb^{\mu},\sa^{\mu}\}^2+\frac{1}{2}\{\ttb^{\mu},\phi^{\ib}\}\{\ttb_{\mu},\phi_{\ib}\}\nn\\
    &+\{\ttb^{\mu},\sa^{\nu}\}\{\sa_{\mu},\sa_{\nu}\}+\{\ttb^{\mu},\phi^{\ib}\}\{\sa_{\mu},\phi_{\ib}\}\nn\\
    &+\frac{1}{4}\{\sa^{\mu},\sa^{\nu}\}\{\sa_{\mu},\sa_{\nu}\}+\frac{1}{2}\{\sa^{\mu},\phi^{\ib}\}\{\sa_{\mu},\phi_{\ib}\}+\frac{1}{4}\{\phi^{\ib},\phi^{\jb}\}\{\phi_{\ib},\phi_{\jb}\}+\ldots\Big)+S_{\BG}
\end{align}
 in terms of Poisson brackets,
where $S_{\BG}$ denotes background action with terms that are zeroth or first orders in fields, and the ellipses denote fermionic terms that we will not consider any further in this paper. The symplectic form $\mho$ takes the form:
\begin{align}
    \mho:=\tK\,\frac{R}{y^4}dy^0dy^1dy^2dy^3=\tK\,(\sinh(\tau))^{-1}d^4y\,,
\end{align}
where $\tK$ is the invariant symplectic form on the internal space-like $S^2$ fiber over $\cM^{1,3}$. The details of this measure will be discussed after introducing the notion of Weyl spinors. Note that \eqref{FLRWaction} is the action that we wish to find a global spinorial description for.

\paragraph{Comments on previous work.} The construction of the $SU(2,2)$ doubleton minimal representation underlying the above construction suggests an organization in terms of the 
compact subgroup $SU(2)_L\times SU(2)_R$ of $SU(2,2)$, as reviewed in Appendix \ref{app:C}. However,
the associated spinors are not very useful for the Lorentzian theory, because $SU(2)_L\times SU(2)_R$ is broken by the background to $SU(2)_{\rm diag}$. This severely restricts the power of the spinor formalism. 
We will overcome this issue by considering Weyl spinors that transform under the space-like isometry group $SL(2,\C)$.

\section{Twistorial realization of 
\texorpdfstring{$\cM^{1,3}$}{M13}} \label{sec:3}
This section provides a spinorial and twistorial realization of the space-time $\cM^{1,3}$. This is non-standard due to the lack of manifest local Lorentz invariance: the FLRW geometry only admits a space-like $SO(1,3) \cong SL(2,\C)$ symmetry. We, therefore, must develop a spinor formalism based on this space-like $SL(2,\C)$ isometry. Nevertheless, our space-like spinorial formalism will be as powerful as the usual one, and the reader should always keep in mind that our spinors are {\em not} the spinors of the local Lorentz group.

\subsection{Spinorial representation and non-compact twistor space}

\paragraph{Chiral basis.} The spinorial representation of $\mso(2,4) \cong \msu(2,2)$ can be defined using the $SO(1,3)$ gamma matrices $\gamma_{\mu}$ obeying the Clifford algebra endowed with the anti-commutators:
\begin{align}
    \{\gamma_{\mu},\gamma_{\nu}\}_+=-2\eta_{\mu\nu}\,,\qquad \mu,\nu=0,1,2,3\,.
\end{align}
To define Weyl spinors transforming under the space-like $SO(1,3)\simeq SL(2,\C)$, we adopt the following chiral basis:\footnote{Notice a slight difference with the chiral basis in \cite{Sperling:2018xrm} (see also \cite{Govil:2013uta}). Very roughly, we swap $\gamma_0$ for $\gamma_4$ so that all $\gamma_{\mu}$ are off-diagonal.}
\begin{align}
\label{chiralso(1,5)basis}
    (\gamma_0)^A{}_{B}=\begin{pmatrix}
     0 & \one \\
     \one & 0
    \end{pmatrix}\,,\quad (\gamma_m)^A{}_{B}=\begin{pmatrix}
     0 & (\sigma_{m})^{\alpha}{}_{\dot\beta}\\
     -(\sigma_m)^{\dot\alpha}{}_{\beta} &0
    \end{pmatrix}\,,\quad (\gamma_4)^A{}_{B}=-\im  \begin{pmatrix}
     -\one  & 0\\
     0 & \one
    \end{pmatrix}\,,
\end{align}
where $\sigma_m$ with $m=1,2,3$ are Pauli matrices, and $\gamma_4:=\,\gamma_0\gamma_1\gamma_2\gamma_3=-\gamma_4^{\dagger}$. Note that 
\begin{align}
    \gamma_a^{\dagger}=\gamma_0\gamma_a\gamma_0^{-1}=-\eta_{ab}\gamma^b\,,\qquad a=0,1,\ldots,4\,.
\end{align}
It is convenient to define the following $\Sigma$ matrices
\begin{align}\label{Sigma-basis}
    \Sigma_{\mu\nu}:=\frac{\im}{4}[\gamma_{\mu},\gamma_{\nu}]\,,\qquad \Sigma_{\mu4}:=\frac{\im }{2}\gamma_{\mu}\gamma_4\,,\qquad \Sigma_{\mu5}:=\frac{1}{2}\gamma_{\mu}\,,\qquad \Sigma_{45}:= \frac 12\gamma_4\,,
\end{align}
which satisfy the $\msu(2,2)$ algebra as well as the reality condition 
\begin{align}
     \label{Sigma-reality-4}
    \Sigma_{\tA\tB}^{\dagger}=\gamma_0 \Sigma_{\tA\tB}\gamma_0^{-1}\,,\qquad \tA,\tB=0,1,\ldots,5\, .
\end{align}
They will be used to construct twistor space and $\cM^{1,3}$ via the Hopf maps defined below. The relevant unitary representations of $SO(2,4)$ are reviewed in Appendix \ref{app:C}.

\paragraph{Twistor space.} The non-compact complex 3-dimensional projective twistor space $\P^{1,2}$ can be defined as a quotient of $H^{3,4}$ (cf. \eqref{eq:H34}) by a $U(1)$ factor, i.e.
\begin{align}
    \P^{1,2}:=H^{3,4}/U(1)\, .
\end{align}
Here $H^{3,4}$ is defined as
\begin{align}\label{eq:H34}
    H^{3,4}\simeq H^7:=\{Z^A\in \C^4\,|\,\bar Z_AZ^A=N\} \ 
\end{align}
where $Z^A$ is an $\msu(2,2)$ vector and
\begin{align}
    \bar{Z}_A:=(Z^B)^{\dagger}(\gamma_0)^B{}_A
\end{align}
is its Dirac conjugate. 
Since $SU(2,2)$ acts transitively on $\P^{1,2}$, we can also view $\P^{1,2}$ as the coset of $SU(2,2)$ modulo the isotropy subgroup,
\begin{align}
    \P^{1,2}\simeq \frac{SU(2,2)}{SU(1,2)\times U(1)}\,.
\end{align}
Remarkably, the above definition of $\P^{1,2}$ coincides with the 6-dimensional co-adjoint orbit of $SU(2,2)\simeq SO(2,4)$. We can therefore identify 
the algebra generated by $y^\mu$ and $\ttb_\mu$ with the  algebra of functions on $\P^{1,2}$. To make this explicit, let us define an  $SU(2,2)$-invariant number operator as
\begin{align}\label{eq:numberoperator}
    \hat\cN:= \bar Z_A Z^A
    =N=\frac{2R}{\ell_p}
    \,,\qquad A=1,2,3,4\, 
\end{align}
anticipating the constraint \eqref{constraint-N} resulting from unitarity of the doubleton representations, cf. Appendix \ref{app:C}. 
Next, we shall equip the twistor space with an $SU(2,2)$--invariant
Poisson structure is given by
\begin{align}\label{eq:commZ}
    \{Z^A,\bar{Z}_B\}= -\im \delta^A{}_B\qquad \text{i.e.}\qquad \{Z^A,\bar Z^B\}=+\im C^{AB}\,,
\end{align}
where $C^{AB}$
is the $\mso(1,4)$-invariant matrix
, cf. Appendix \ref{app:C}. 
Then the algebra of functions on $\P^{1,2}$ can be defined as
\begin{equation}\label{fuzzyP12}
    \begin{split}
    C^{\infty}(\P^{1,2}_N):=\Big\{f(Z,\bar{Z})\in \C[Z^A,\bar{Z}_B]\mod U(1)\ \Big|\ \{\hat\cN,f(Z,\bar{Z})\}=0 \Big\}\, .
    \end{split}
\end{equation}
Observe that $\hat \cN$ 
defines a gradation on the Poisson algebra generated by the $Z^A$ via
\begin{align}\label{eq:gradingZ}
    \{\hat{\cN},Z^A\}=+\im Z^A\,,\qquad \{\hat{\cN},\bar{Z}_A\}=-\im \bar{Z}_A\,.
\end{align}
Hence $C^{\infty}(\P^{1,2})$ consists of polynomial functions that have equal numbers of $Z$'s and $\bar Z$'s.

\paragraph{Weyl spinors.} Let us now parametrize the $\msu(2,2)$-vectors $Z^A$, which are homogeneous coordinates on $\P^{1,2}$, in terms of two spinors transforming in the fundamental of the non-compact space-like isometry subgroup $SL(2,\C)\subset SU(2,2)$
\begin{align}
    Z^A=\binom{\lambda^{\alpha}}{\mu^{\dot\alpha}}\,,\qquad\quad \text{hence}\qquad \bar Z_A=(\bar\mu_{\alpha},\bar\lambda_{\dot\alpha})\,.
\end{align}
In contrast with the compact case considered in \cite{Sperling:2018xrm,Steinacker:2023zrb} based on $SU(2)_L\times SU(2)_R$, all spinors associated with the chiral basis \eqref{chiralso(1,5)basis} are now \emph{Weyl spinors} of $SL(2,\C)$. In terms of these Weyl spinors, the number operator $\hat\cN$ reads
\begin{align}\label{eq:Ninspinors}
     \hat\cN=\langle \lambda\,\bar\mu\rangle+[\mu\,\bar\lambda]=N\,,
\end{align}
where $\langle u \,v\rangle := u^{\alpha}v_{\alpha}\,,\ [u\,v]:=u^{\dot\alpha}v_{\dot\alpha}$. The non-vanishing Poisson brackets of the spinors take the form 
\begin{align}\label{eq:Poissonspinors}
    \{\lambda^{\alpha},\bar\mu^{\beta}\}=+\im\eps^{\alpha\beta}\,,\qquad \{\mu^{\dot\alpha},\bar\lambda^{\dot\beta}\}=+\im\eps^{\dot\alpha\dot\beta}\,,
\end{align}
and 
\begin{align}
    \{\cN,\lambda^{\alpha}\}=+\im\lambda^{\alpha}\,,\quad \{\cN,\mu^{\dot\alpha}\}=+\im\mu^{\dot\alpha}\,,\quad \{\cN,\bar\mu_{\alpha}\}=-\im\bar\mu_{\alpha}\,,\quad \{\cN,\bar\lambda_{\dot\alpha}\}=-\im\bar\lambda_{\dot\alpha}\,.
\end{align}
The algebra of functions $C^{\infty}(\P^{1,2})$ reads in terms of Weyl spinors as follows:
\begin{align}\label{operator-algebra-1}
    C^{\infty}(\P^{1,2})=\sum_{n+m=p+q}f_{\alpha(n)\dot\alpha(m)}{}^{\beta(p)\dot\beta(q)}\lambda^{\alpha(n)}\mu^{\dot\alpha(m)}\bar{\mu}_{\beta(p)}\bar{\lambda}_{\dot\beta(q)}\,.
\end{align}

\paragraph{Hopf map from $\P^{1,2}$ to $H^4$.} Using the above data, we can then define generators $m_{\tA\tB}$ of $SO(2,4)\simeq SU(2,2)$ via the Jordan-Schwinger construction as
\begin{align}\label{m-ZZ}
    m_{\tA\tB}:=\bar{Z}_A(\Sigma_{\tA\tB})^A{}_BZ^B\,,\qquad \tA,\tB=0,1,\ldots,5\, .
\end{align}
Note that these are hermitian due to \eqref{Sigma-reality-4}, and hence can be used to define a unitary representation cf. Appendix \ref{app:C}. From \eqref{m-ZZ} 
we can define the following hermitian generators
\begin{subequations}\label{Hopf-Z}
    \begin{align}
    D:&= \frac 12\bar Z_A(\gamma_4)^A{}_BZ^B\, = \frac{1}{\ell_p} y_4   ,\\
    y_a:&=\frac{\ell_p}{2}\bar{Z}_A(\gamma_a)^A{}_BZ^B\,,\qquad\qquad\qquad\qquad\qquad\ \ \  a=0,1,2,3,4\,,\\
    \ttb_{\mu}:&=\frac{1}{R}\bar{Z}_A(\Sigma_{\mu 4})^A{}_BZ^B \ 
    = \frac{1}{R}\bar{Z}_A(\underline{\gamma}_\mu)^A{}_BZ^B 
    \,,\qquad \mu=0,1,2,3\,. \label{eq:Hopf-t}
\end{align}
\end{subequations}
In particular, we can interpret $y^a$ with $y_a y^a = - R^2$ geometrically in terms of the Hopf map 
\begin{align}
y_a:\quad H^{3,4} &\rightarrow H^4 \ \subset \R^{1,4} \, .
\end{align}
The gamma matrices $\underline{\gamma}_{\mu} :=\Sigma_{\mu 4}$ are given explicitly by
\begin{align}\label{gammaSO(1,3)}
    (\underline{\gamma}_0)^{A}_{\ B}=\frac 12 \begin{pmatrix}
     0 &\one_2\\
     -\one_2 & 0
    \end{pmatrix}\,,\quad  (\underline{\gamma}_{m})^A_{\ B}\equiv (\underline{\gamma}_m)^A_{\ B}=\frac 12  \begin{pmatrix}
     0 &\sigma_m\\
     \sigma_m & 0
    \end{pmatrix}\,,\qquad m=1,2,3\,
\end{align}
which can be subsumed as
\begin{align}
    (\underline{\gamma}_{\mu})^A{}_B=\frac{1}{2}\begin{pmatrix}
        0& \sigma_{\mu}\\
        -\bar\sigma_{\mu} &0
    \end{pmatrix}\,,\qquad \sigma_{\mu}=(1,\sigma_m)\,,\qquad \bar\sigma_{\mu}=(1,-\sigma_m)\,.
\end{align}

\paragraph{Hopf map from $\P^{1,2}$ to $\cM^{1,3}$.} We now illustrate how our new construction terms of Weyl spinors lead to a global description of coordinate functions $\ttb_\mu$ and $y^\mu$. The Hopf maps \eqref{Hopf-Z} relating $(Z,\bar Z)$ with the triplet $(D,y_{\mu},\ttb_{\mu})$ in $3+1$ dimensions can be written as
\begin{subequations}\label{Hopf-Weyl}
    \begin{align}
    D&=\frac{1}{2}\bar Z_A(\gamma_4)^A{}_BZ^B=-\frac{\im}{2}\Big(-\langle \lambda\,\bar\mu\rangle + [\mu\,\bar\lambda]\Big)\,,\label{D-Weyl-1}\\
    y_\mu&=\frac{\ell_p}{2}\bar{Z}_A(\gamma_\mu)^A{}_BZ^B=\frac{\ell_p}{2}\Big(\bar\mu_{\alpha}(\sigma_{\mu})^{\alpha}{}_{\dot\alpha}\mu^{\dot\alpha}+\bar\lambda_{\dot\alpha}(\bar\sigma_{\mu})^{\dot\alpha}{}_{\alpha}\lambda^{\alpha}\Big)\,,\label{y-pre-incidence}\\
    \ttb_{\mu}&=\frac{1}{R}\bar{Z}_A(\Sigma_{\mu 4})^A{}_BZ^B=\frac{1}{2R}\Big(\bar\mu_{\alpha}(\sigma_{\mu})^{\alpha}{}_{\dot\alpha}\mu^{\dot\alpha}-\bar\lambda_{\dot\alpha}(\bar\sigma_{\mu})^{\dot\alpha}{}_{\alpha}\lambda^{\alpha}\Big)\,, \label{t-pre-incidence} 
\end{align}
\end{subequations}
 recalling $\sigma_{\mu}=(\one,\sigma_m)$ and $\bar\sigma_{\mu}=(\one,-\sigma_m)$ for $m=1,2,3$. 
 It is not hard to verify that these generators satisfy all the relations \eqref{M31-rlations}.
 These expressions provide a definition of 
$\ttb_{\mu}$ in terms of Weyl spinors transforming under the space-like $SO(1,3)$ isometry, which should not be confused with the local Lorentz group. 
For later use, we also note the following spinorial expression the ``missing" generator $y^4$:
\begin{align}
\label{y4-spinor}
    y_4=\frac{\ell_p}{2}\bar{Z}_A(\gamma_4)^A{}_BZ^B=-\im\frac{\ell_p}{2}\Big(-\langle \lambda\,\bar\mu\rangle +[\mu\,\bar\lambda]\Big)=\frac{N\ell_p}{2}\sinh(\tau)=R\sinh(\tau)\, .
\end{align}


\paragraph{Incidence relations.} 
We can now eliminate all right-handed Weyl spinors $\mu^{\dot\alpha}$ and $\bar\lambda_{\dot\alpha}$ using the incidence relations as in \cite{kodaira1962theorem,Penrose:1967wn}. These incidence relations can be derived directly from our definition of the Hopf maps. In particular, from \eqref{y-pre-incidence} and \eqref{t-pre-incidence} we obtain
\begin{subequations}\label{y-t-spinors}
    \begin{align}
    y^{\alpha\dot\alpha}:&=y^{\mu}\sigma_{\mu}^{\alpha\dot\alpha}=+\ell_p(\lambda^{\alpha}\bar\lambda^{\dot\alpha}-\bar\mu^{\alpha}\mu^{\dot\alpha})\,,\\
    \ttb^{\alpha\dot\alpha}:&=\ttb^{\mu}\sigma_{\mu}^{\alpha\dot\alpha}=-\frac{1}{R}(\lambda^{\alpha}\bar\lambda^{\dot\alpha}+\bar\mu^{\alpha}\mu^{\dot\alpha})\,,
\end{align}
\end{subequations}
where $\sigma_{\mu}^{\alpha\dot\alpha}\sigma^{\mu}_{\beta\dot\beta}=2\epsilon^{\alpha}{}_{\beta}\epsilon^{\dot\alpha}{}_{\dot\beta}$ and $\bar\sigma_{\mu}^{\alpha\dot\alpha}\sigma^{\mu}_{\beta\dot\beta}=-2\epsilon^{\alpha}{}_{\beta}\epsilon^{\dot\alpha}{}_{\dot\beta}$ 
(where $\bar{\sigma}^{\dot\alpha\alpha}\equiv\sigma^{\alpha\dot\alpha}$). 
It can be checked that these satisfy the reality conditions
\begin{align}
\label{y-t-spinor-real}
    (y^{\alpha\dot\beta})^* = y_{\beta\dot\alpha}\,, 
    \qquad  (\ttb^{\alpha\dot\beta})^* = \ttb_{\beta\dot\alpha}\,.
\end{align}
Furthermore, we can solve these for $(\mu,\bar\lambda)$ as
\begin{subequations}\label{incidence-relations}
    \begin{align}
    \mu^{\dot\alpha}&=\frac{1}{\ell_p\langle \lambda\,\bar\mu\rangle}y^{\alpha\dot\alpha}\lambda_{\alpha}=+\frac{R}{\langle \lambda\,\bar\mu\rangle}\ttb^{\alpha\dot\alpha}\lambda_{\alpha}\,,\\
    \bar\lambda^{\dot\alpha}&=\frac{1}{\ell_p\langle \lambda\,\bar\mu\rangle }y^{\alpha\dot\alpha}\bar\mu_{\alpha}=-\frac{R}{\langle \lambda\,\bar\mu\rangle }\ttb^{\alpha\dot\alpha}\bar\mu_{\alpha}\,.
\end{align}
\end{subequations}
Using these incidence relations, we can recast the Poisson algebra \eqref{operator-algebra-1} as 
\begin{equation}\label{operator-algebra-2}
    \begin{split}
    C^{\infty}(\P^{1,2}):=\Big\{f(y|\lambda,\bar{\mu})\in \C[\lambda,\bar{\mu}]\ \Big|\ \{\hat\cN,f(y|\lambda,\bar{\mu})\}=0 \Big\}\,.
    \end{split}
\end{equation}
In order words, 
\begin{align}
    C^{\infty}(\P^{1,2})=\sum_{s}f_{\alpha(s)}{}^{\beta(s)}(y)\,\lambda^{\alpha(s)}\bar\mu_{\beta(s)}\,.
\end{align}
Note that the coefficients $f_{\alpha(s)}{}^{\beta(s)}$ have {\em only un-dotted indices}, which means that they are chiral field representations. 

\paragraph{Remark.} We observe that although the twistor action of Lorentzian \hsikkt is real, its spacetime fields (the coefficients $f$'s after integrating out fiber coordinates cf. Appendix \ref{app:A}) will typically be complex. This is not in contradiction with unitarity or hermiticity of the model;
for example, GR with a cosmological constant can also be formulated in terms of chiral field representations; see e.g. \cite{Krasnov:2016emc}.\footnote{See also \cite{Basile:2022mif} for the construction of partially massless chiral higher-spin theories using chiral field representations.}

\subsection{Bundle structure of \texorpdfstring{$\cM^{1,3}$}{M13}}\label{sec:bundle}
Recall that $\P^{1,2}$ can be understood as $SO(1,4)$-equivariant bundle over $H^4$ cf. \cite{Sperling:2018xrm} via the Hopf map 
\begin{align}
    \pi:\P^{1,2}\simeq H^{3,4}/U(1)\rightarrow H^4
\end{align}
which is compatible with $SO(1,4)$.
We can choose a reference point $\tp\in H^4$
 using a suitable $SO(1,3)$ transformation such that
\begin{align}
    y^a\big|_{\tp}=\big(R \cosh(\tau),0,0,0,R\sinh(\tau)\big)\,.
\end{align}
The stabilizer of $\tp$ is $SO(3)$, which is part of the isometry group $SO(1,3)$ on $\cM^{1,3}$. 

Now consider the twistor $Z^A$.
Upon applying an appropriate $SL(2,\C)$ transformation, we can map $Z^A$ to a reference twistor $\mathtt{Z}^A$ over $\tp$ where 
\begin{align}
   Z^A\mapsto \mathtt{Z}^A:=Z^A\big|_{\tp}=\begin{pmatrix}
        \lambda \\ \mu 
    \end{pmatrix}\, \quad \mbox
    {with} \ \ \mu = e^{\im \vartheta }\lambda \ ,
\end{align}
i.e. the two Weyl spinors $\lambda$ and $\mu$ coincide up to a phase. This follows by observing that $y^i=0$ for $i=1,2,3$ and $\ttb_0 = 0$ at the reference point $\tp$, which translates into
\begin{subequations}
    \begin{align}
    y_i\big|_{\tp} &=\frac{\ell_p}{2}\Big(\bar\mu_{\alpha}(\sigma_{i})^{\alpha}{}_{\dot\alpha}\mu^{\dot\alpha}-\bar\lambda_{\dot\alpha}(\sigma_i)^{\dot\alpha}{}_{\alpha}\lambda^{\alpha}\Big) = 0 \,,\\ 
   \ttb_0\big|_{\tp} &=\frac{\ell_p}{2}\Big(\bar\mu_{\alpha}(\sigma_{0})^{\alpha}{}_{\dot\alpha}\mu^{\dot\alpha} -\bar\lambda_{\dot\alpha}(\sigma_0)^{\dot\alpha}{}_{\alpha}\lambda^{\alpha}\Big) = 0\,.  
\end{align}
\end{subequations}
\normalsize
where we recall that $\sigma_{\mu}=(\one,\sigma_m)$ and $\bar\sigma_{\mu}=(\one,-\sigma_m)$. 
Hence, $\mu = e^{\im  \vartheta}\lambda$ as claimed. This 
exhibits the structure of $\P^{1,2}$ as a $S^2$-bundle over $\cM^{1,3}$ in spinorial language, complementing the vectorial description in \eqref{eq:orthogonalofPY} and \eqref{S2sphereM13}. 



\section{Spinorial representation of the Lorentzian \texorpdfstring{\hsikkt}{hsikkt} theory}\label{sec:SO(1,3)spinors-IKKT}

This section studies the higher-spin gauge theory on $\cM^{1,3}$ defined by the matrix background \eqref{BG-fields}. First, we show that the effective metric can be nicely expressed in terms of Weyl spinors. This allows us to derive the spinorial representation for the spacetime action of Lorentzian \hsikkt theory on $\cM^{1,3}$ in the flat limit.


\subsection{Spinorial background geometry}

We now consider the background $\ttb^{\mu}$ cf. \eqref{BG-fields}
of the matrix model, and elaborate on the associated frame and metric in the present spinor formalism. 

\paragraph{Effective vielbein and metric.} We have shown that $(\ttb^{\mu},y^{\nu})$ can be defined globally in terms of Weyl spinors, cf. \eqref{Hopf-Weyl}. This allows to compute the following background derivative
\begin{align}\label{derivation-def}
    \{\ttb^{\mu},\varphi(y|\lambda,\bar\mu)\}&=\Big(\{\ttb^{\mu},y^{\nu}\}\frac{\p}{\p y^{\nu}}+\{\ttb^{\mu},\lambda^{\alpha}\}\frac{\partial}{\p \lambda^{\alpha}}+\{\ttb^{\mu},\bar\mu^{\alpha}\}\frac{\partial}{\p \bar\mu^{\alpha}}\Big)\varphi(y|\lambda,\bar\mu)\nn\\
    &=:\Big(E^{\mu\nu}\frac{\p}{\p y^{\nu}}+E^{\mu|\alpha}\frac{\p}{\p \lambda^{\alpha}}+\bar{E}^{\mu|\alpha}\frac{\p}{\p\bar\mu^{\alpha}}\Big)\varphi(y|\lambda,\bar\mu)\,,
\end{align}
for any $\hs$-valued function $\varphi(y|\lambda,\bar\mu)\in C^{\infty}(\P^{1,2})$. Here the symbols $E^{\mu|\bullet}$ are referred to as the effective vielbein. 
Using the Hopf maps \eqref{Hopf-Weyl}, one can show that 
\begin{align}
    \{\ttb^{\mu},y^{\nu}\}&=-\frac{\ell_p}{4R}\Big[\big\{\bar\lambda_{\dot\alpha}(\bar\sigma^{\mu})^{\dot\alpha}{}_{\beta}\lambda^{\beta},\bar\mu_{\gamma}(\sigma^{\nu})^{\gamma}{}_{\dot\delta}\mu^{\dot\delta}\big\}+\big\{\bar\lambda_{\dot\gamma}(\bar\sigma^{\nu})^{\dot\gamma}{}_{\delta}\lambda^{\delta},\bar\mu_{\alpha}(\sigma^{\mu})^{\alpha}{}_{\dot\beta}\mu^{\dot\beta}\big\}\Big]\nn\\
    &=+\im \frac{\ell_p}{4R}\Big(\lambda^{\alpha}\big[(\bar\sigma^{\mu})_{\dot\alpha\alpha}(\sigma^{\nu})^{\gamma\dot\alpha}+(\bar\sigma^{\nu})_{\dot\alpha\alpha}(\sigma^{\mu})^{\gamma\dot\alpha}\big]\bar\mu_{\gamma}\nn\\
    &\qquad \qquad \qquad \qquad \qquad \qquad -\bar\lambda_{\dot\alpha}\big[(\bar\sigma^{\mu})^{\dot\alpha\gamma}(\sigma^{\nu})_{\gamma\dot\gamma}+(\bar\sigma^{\nu})^{\dot\alpha\gamma}(\sigma^{\mu})_{\gamma\dot\gamma}\big]\mu^{\dot\gamma}\Big)\nn\\
    &=-\im \frac{\ell_p}{2R}\big([\mu\,\bar\lambda]-\langle \lambda\,\bar\mu\rangle \big)\eta^{\mu\nu}\,,
\end{align}
\normalsize
where we have used 
\begin{subequations}
\begin{align}
    (\sigma^{\mu})_{\alpha\dot\beta}(\bar\sigma^{\nu})^{\dot\beta\gamma}+(\sigma^{\nu})_{\alpha\dot\beta}(\bar\sigma^{\mu})^{\dot\beta\gamma}=-2\eta^{\mu\nu}\epsilon_{\alpha}{}^{\gamma}\,,\\
    (\bar\sigma^{\mu})_{\dot\alpha\gamma}(\sigma^{\nu})^{\gamma\dot\beta}+(\bar\sigma^{\nu})_{\dot\alpha\gamma}(\sigma^{\nu})^{\gamma\dot\beta}=-2\eta^{\mu\nu}\epsilon_{\dot\alpha}{}^{\dot\beta}\,.
\end{align}
\end{subequations}
Since $\eta^{\mu\nu}(\sigma_{\mu})^{\alpha\dot\alpha}(\sigma_{\nu})^{\beta\dot\beta}=2\eps^{\alpha\beta}\eps^{\dot\alpha\dot\beta}$, we can also write the above as
\begin{align}
    \{\ttb^{\alpha\dot\alpha},y^{\beta\dot\beta}\}
    =-\im\frac{\ell_p}{R} \big([\mu\,\bar\lambda]-\langle \lambda\,\bar\mu\rangle \big)\eps^{\alpha\beta}\eps^{\dot\alpha\dot\beta}
=\sinh(\tau)\eps^{\alpha\beta}\eps^{\dot\alpha\dot\beta}\,,
\end{align}
which coincides with \eqref{eq:PoissonM13} since
\begin{align}\label{bracket-relations-1}
    -\im(-\langle \lambda\,\bar\mu\rangle+[\mu\,\bar\lambda])
    =D 
    =N\sinh(\tau)
\end{align}
using \eqref{y4-spinor}. Given that $\langle \lambda\,\bar\mu\rangle+[\mu\,\bar\lambda]=N$, cf. \eqref{eq:Ninspinors}, we have
\begin{align}\label{inner-prod-spinors}
    \langle \lambda\,\bar\mu\rangle =\frac{N}{2}\big[1-\im \sinh(\tau)\big]\,,\qquad [\mu\,\bar\lambda]=\frac{N}{2}\big[1+\im \sinh(\tau)\big]\,.
\end{align}
This reflects the space-like time-dependent nature of the internal $S^2$ that encodes higher-spin structures. 
In particular,
the effective metric can be read off by considering the following kinetic term 
\begin{align}\label{effectivemetricFLRW}
    \{\ttb^{\mu},\varphi(y)\}\{\ttb_{\mu},\varphi(y)\}:=\gamma^{\mu\nu}\p_{\mu}\varphi(y)\p_{\nu}\varphi(y)\,,\qquad \gamma^{\mu\nu}=\eta^{\mu\nu}\sinh^2(\tau)\,
\end{align}
with $\varphi(y)$ a scalar field,  thus recovering the geometry \cite{Sperling:2019xar}.
In terms of spinors, the effective metric reads\footnote{The full effective metric is obtained by taking into account an appropriate dilatation, cf. \cite{Steinacker:2019fcb}.}
\begin{align}\label{eq:eff-metric-spinor}
    \gamma^{\alpha\dot\alpha\beta\dot\beta}=2\epsilon^{\alpha\beta}\epsilon^{\dot\alpha\dot\beta}\sinh^2(\tau)\,.
\end{align}
As a result, we recovered the ``standard'' description for the $\Box:=\frac{1}{2}\p^{\alpha\dot\alpha}\p_{\alpha\dot\alpha}$ operator. This resolves the obstruction in \cite{Steinacker:2023zrb}  for a global description for Lorentzian \hsikkt based on $SU(2)_L\times SU(2)_R\subset SU(2,2)$.\footnote{Intriguingly, the Lorentzian twistor approach present here should also complement the standard twistor construction for higher-spin theories which relies heavily on spacetime with Euclidean or split signature, see e.g. \cite{Tran:2022tft,Herfray:2022prf,Adamo:2022lah}.}




\subsection{Higher-spin degrees of freedom}
Armed with the above organization of higher-spin fields based on $\P^{1,2}$, we can count their physical degrees of freedom. 
Since counting degrees of freedom for $\hs$-valued functions $\varphi(y|\lambda,\bar\mu)$ is straightforward, we focus 
 on the gauge potential $\sa_{\mu}$. 
 First, we can express the real $\hs$-valued gauge field $\sa_\mu^* = \sa_\mu$ in spinorial form $\sa_{\alpha\dot\beta} := \sa_{\mu}\sigma^{\mu}_{\alpha\dot\beta}$, which satisfies the reality condition
 \begin{align}
 \label{a-spinor-reality}
    (\sa_{\alpha\dot\beta})^* =   \sa^{\beta\dot\alpha} \ .
 \end{align}
For the purpose of computing interaction vertices, it is advantageous to write these spinorial gauge fields in a chiral form as in \eqref{operator-algebra-2}, which leads to the following mode decomposition:
\begin{align}\label{mode-expansion-t-lambda}
\sa_{\alpha\dot\alpha}=\sum_{s=0}\cA_{\beta(2s)|\alpha\dot\alpha}\lambda^{\beta(s)}\bar\mu^{\beta(s)}\,,\qquad \cA_{\beta(2s)|\alpha\dot\alpha}:=A_{(\beta(2s)\alpha)\,\dot\alpha}
+4s\,\eps_{\alpha\beta}\sA_{\beta(2s-1)\,\dot\alpha}\,,
\end{align}
with a total of $\sum_s 2(2s+2) + 2(2s) = \sum_s 4(2s+1)$ off-shell degrees of freedom,
as in the vectorial description. Here the coefficients $A_{\alpha(2s-1)\,\dot\alpha}$ and $\sA_{\alpha(2s-3)\,\dot\alpha}$ are tensorial fields that are in a sense chiral, since they have more un-dotted than dotted indices.
They also satisfy certain reality conditions resulting from \eqref{a-spinor-reality}, which we shall not write down explicitly.
Note that we have adjusted the prefactors of $\sA_{\beta(2s-1)\,\dot\alpha}$ to be $4s$ for later convenience.

\paragraph{Gauge fixing and dof.} 
As usual, we
 impose the gauge fixing condition for Yang-Mills matrix models \cite{Blaschke:2011qu}
\begin{align}
\cG(\sa) = \{\ttb_{\alpha\dot\alpha},\sa^{\alpha\dot\alpha}\}=0 \ 
 \label{gaugefix-intertwiner}
 \end{align}
to remove $2s+1$ dof,
and use the linearized gauge transformation
\begin{align}
    \delta_{\xi} \sa^{\mu}=\{\ttb^{\mu},\xi\}\,,\qquad \xi=\xi(y|\lambda,\bar\mu)\in C^{\infty}(\P^{1,2})
\end{align}
to remove further $2s+1$ on-shell pure gauge components\footnote{Note that $\hs$-valued gauge fields in \hsikkt are {\em not} divergent-free $\p^{\alpha\dot\alpha}\sa_{\alpha\dot\alpha} \neq 0$ a priori. 
}. This leaves us with $2(2s+1)$ physical dof for the ``would-be massive'' Lorentzian $\hs$-valued potential $\cA_{\beta(2s)|\alpha\dot\alpha}$. Hence, a gauge field with internal spin $s$ in \hsikkt carries $2(2s+1)$ degrees of freedom.
This falls in between a $4d$ massive higher-spin theory with $2s+1$ dof and a $4d$ conformal higher-spin gravity with $s(s+1)$ dof. However, we show below that after diagonalizing higher-spin modes and doing gauge-fixing in spacetime appropriately, we can get the correct degrees of freedom for a massive higher-spin field, which is $2s+1$.



\subsection{Spacetime action of Lorentzian \texorpdfstring{\hsikkt}{hsikkt} in the flat limit}\label{sec:action}
Having established all the necessary ingredients, we shall discuss the spacetime action of \hsikkt on $\cM^{1,3}$ in the local 4-dimensional regime, where the geometry can be considered as flat \cite{Steinacker:2023zrb}. We will focus only on the Yang-Mills (YM) part of the action.
 \paragraph{Local 4-dimensional regime.} 
To obtain a simple description of the local physics, we consider the local 4-dimensional regime following \cite{Steinacker:2023cuf}. This is the regime where the wavelength of a function $\varphi(y)$ is much shorter than the cosmic scale but much longer than the scale of non-commutativity. In that regime, we can treat any $\hs$-valued function as a $4d$ function, and the $\hs$ generators $\ttb$ as commutative variables. In particular, the local geometry is effectively flat. Then the contributions of the Poisson brackets acting on fiber coordinates $(\lambda,\bar\mu)$ can be neglected, i.e. we will drop terms that are proportional to $E^{\mu|\alpha}$ or $\bar{E}^{\mu|\alpha}$ cf. \eqref{derivation-def}. Then the spinorial form of the vielbein reduces to
\begin{align}
    E^{\mu\nu}=\eta^{\mu\nu}\sinh(\tau)\mapsto E^{\alpha\dot\alpha,\beta\dot\beta}= 2\epsilon^{\alpha\beta}\epsilon^{\dot\alpha\dot\beta}\sinh(\tau)\,,
\end{align}
and the following approximation of the Poisson brackets:
\begin{align}
\{f,g\} &\approx \theta^{\mu\nu} \del_\mu f\, \del_\nu g \ \propto \theta^{\alpha\dot\alpha\beta\dot\beta}\p_{\alpha\dot\alpha}f\,\p_{\beta\dot\beta}g
 \label{poisson-approx-x}
\end{align}
since $E^{\mu|\alpha}$ and $\bar{E}^{\mu|\alpha}$ are sub-leading in the flat limit. In terms of spinors, the Poisson structures read
\begin{align}\label{Poisson-spinor-1}
    \theta^{\mu\nu}\mapsto \theta^{\alpha\dot\alpha\,\beta\dot\beta}:=\{y^{\alpha\dot\alpha},y^{\beta\dot\beta}\}=-\im\ell_p^2\Big[\big(\lambda^{\alpha}\bar\mu^{\beta}+\lambda^{\beta}\bar\mu^{\alpha}\big)\epsilon^{\dot\alpha\dot\beta}+\big(\mu^{\dot\alpha}\bar\lambda^{\dot\beta}+\mu^{\dot\beta}\bar\lambda^{\dot\alpha}\big)\epsilon^{\alpha\beta}\Big]\,
\end{align}
which has both left-handed (un-dotted) and right-handed (dotted) Weyl spinors. Note that this Poisson structure can be simplified by using the following identities
\begin{align}
\label{epsilon-l-mubar}
    \lambda^{\alpha}\bar\mu^{\beta}-\lambda^{\beta}\bar\mu^{\alpha}=-\langle \lambda\,\bar\mu\rangle\epsilon^{\alpha\beta} \,,\qquad \mu^{\dot\alpha}\bar\lambda^{\dot\beta}-\mu^{\dot\beta}\bar\lambda^{\dot\alpha}=-[\mu\,\bar\lambda]\epsilon^{\dot\alpha\dot\beta}\,.
\end{align}
We obtain
\begin{align}
 \theta^{\alpha\dot\alpha\beta\dot\beta}
 = -\im\ell_p^2\Big(2\lambda^{\alpha}\bar\mu^{\beta}\eps^{\dot\alpha\dot\beta}
+2\mu^{\dot\alpha}\bar\lambda^{\dot\beta}\eps^{\alpha\beta}
+N\eps^{\alpha\beta}\epsilon^{\dot\alpha\dot\beta}\Big)\, \nn\\
\end{align}
using \eqref{inner-prod-spinors}.
At late times, this reduces to 
\begin{align}\label{theta-at-large-tau}
    \lim_{\tau\rightarrow\infty}\theta^{\alpha\dot\alpha\beta\dot\beta}=-2\im\,\ell_p^2\Big(\lambda^{\alpha}\bar\mu^{\beta}\eps^{\dot\alpha\dot\beta}+\mu^{\dot\alpha}\bar\lambda^{\dot\beta}\eps^{\alpha\beta}\Big)\,.
\end{align}
This is also justified by observing that the expression of $\theta^{\mu\nu}$ cf. \eqref{mgenerator} in the flat limit 
\begin{align}
    \theta^{\mu\nu}=-\ell_p^2m^{\mu\nu}\approx \frac{\ell_p^2\sinh(\tau)}{\cosh^2(\tau)}\big[y^{\mu}\ttb^{\nu}-y^{\nu}\ttb^{\mu}\big]\,,
\end{align}
which can be brought into the form \eqref{theta-at-large-tau} by contracting with $\sigma_{\mu}$'s matrices and take the late-time limit appropriately. 

Since the Poisson structure will appear in the cubic and quartic interactions on twistor space $\P^{1,2}\simeq \P^1\times \cM^{1,3}$, it is desirable to rewrite \eqref{Poisson-spinor-1} in terms of either left-handed or right-handed spinors only. This is achieved using the incidence relations \eqref{incidence-relations}, so that
\begin{align}\label{Poisson-theta-spinor}
    \theta^{\alpha\dot\alpha\beta\dot\beta}=-\im\,\ell_p^2\Big[\big(\lambda^{\alpha}\bar\mu^{\beta}+\lambda^{\beta}\bar\mu^{\alpha}\big)\epsilon^{\dot\alpha\dot\beta}+\frac{1}{\ell_p^2\langle \lambda\,\bar\mu\rangle^2}\big(y^{\gamma\dot\alpha}y^{\delta\dot\beta}+y^{\gamma\dot\beta}y^{\delta\dot\alpha}\big)\lambda_{\gamma}\bar\mu_{\delta}\epsilon^{\alpha\beta}\Big]\,.
\end{align}
In this form, the appearance of 
 explicitly Lorentz-violating structures\footnote{Recall that 
we use the organization in terms of the unbroken $SO(1,3)$ isometry group, while local Lorentz invariance is broken by the background.} 
 (notably the time-like vector field $\cY=y^{\mu}\p_{\mu}$ in the spacetime vertices of Lorentzian \hsikkt, cf. \cite{Steinacker:2023cuf}) is manifest.


\paragraph{Normalization.} To bring our spacetime expressions to the same  canonical normalization as in \cite{Steinacker:2023cuf}, we shall rescale 
\begin{align}\label{canonical-normalization}
    \lambda\mapsto \sqrt{R/\ell_p}e^{\frac{\tau}{2}}\,\lambda\sim \sqrt{\frac{N}{2}}e^{\frac{\tau}{2}}\lambda\,,\qquad \qquad \bar\mu \mapsto \sqrt{R/\ell_p}e^{\frac{\tau}{2}}\,\bar\mu\sim \sqrt{\frac{N}{2}}e^{\frac{\tau}{2}}\bar\mu\,,
\end{align}
such that $\langle \lambda\,\bar\mu\rangle=-1$ and $[\mu\,\bar\lambda]=+1$. This rescaling is suggestive, since $\lambda,\bar\mu$ may be thought of as ``square-root'' of the coordinates functions $y^a$. While this rescaling does not change the length dimension of spinors, it leads to
\begin{align}\label{eq:theta-normalized}
    \theta^{\alpha\dot\alpha\beta\dot\beta}&=-\im R\,\ell_pe^{\tau}\Big[(\lambda^{\alpha}\bar\mu^{\beta}+\lambda^{\beta}\bar\mu^{\alpha})\eps^{\dot\alpha\dot\beta}+(y^{\gamma\dot\alpha}y^{\delta\dot\beta}+y^{\gamma\dot\beta}y^{\delta\dot\alpha})\lambda_{\gamma}\bar\mu_{\delta}\eps^{\alpha\beta}\Big]\nn\\
    &\approx -\im L_{\rm NC}^2\Big[(\lambda^{\alpha}\bar\mu^{\beta}+\lambda^{\beta}\bar\mu^{\alpha})\eps^{\dot\alpha\dot\beta}+(y^{\gamma\dot\alpha}y^{\delta\dot\beta}+y^{\gamma\dot\beta}y^{\delta\dot\alpha})\lambda_{\gamma}\bar\mu_{\delta}\eps^{\alpha\beta}\Big]\,,
\end{align}
where $L_{\rm NC}=\sqrt{R\,\ell_p\cosh(\tau)}$ is the scale of non-commutativity, i.e. the scale set by the commutators of (higher-spin) fields in the semi-classical approximation. As a result, we obtain the same time-dependent coupling constants as in \cite{Steinacker:2023cuf}. 
Below, we exhibit the spacetime action of Lorentzian \hsikkt as well as the result for 3-point scattering amplitudes using this normalization. Note also that we will rescale
\begin{align}
    \theta\mapsto\frac{1}{\langle\lambda\,\bar\mu\rangle}\theta
\end{align}
to match the form of the integral over fiber coordinates $(\lambda,\bar\mu)$ of the internal 2-sphere $\P^1\sim S^2$ as explained in Appendix \ref{app:B}.

\paragraph{Twistor action.} For simplicity, we will consider only the bosonic part of the action \eqref{FLRWaction}. The spinorial description of the twistor action for Lorentzian \hsikkt is
\begin{align}\label{FLRWaction-Weyl}
    S=\int &\mho\,\Big(\frac 12 \{\ttb^{\alpha\dot\alpha},\sa^{\beta\dot\beta}\}\{\ttb_{\alpha\dot\alpha},\sa_{\beta\dot\beta}\}+\{\ttb_{\alpha\dot\alpha},\ttb_{\beta\dot\beta}\}\{\sa^{\alpha\dot\alpha},\sa^{\beta\dot\beta}\}+\frac{1}{2}\{\ttb_{\alpha\dot\alpha},\sa^{\alpha\dot\alpha}\}^2\nn\\
    &+\frac{1}{2}\{\ttb^{\alpha\dot\alpha},\phi^{\ib}\}\{\ttb_{\alpha\dot\alpha},\phi_{\ib}\}+\{\ttb^{\alpha\dot\alpha},\sa^{\beta\dot\beta}\}\{\sa_{\alpha\dot\alpha},\sa_{\beta\dot\beta}\}+\{\ttb^{\alpha\dot\alpha},\phi^{\ib}\}\{\sa_{\alpha\dot\alpha},\phi_{\ib}\}\nn\\
    &+\frac{1}{4}\{\sa^{\alpha\dot\alpha},\sa^{\beta\dot\beta}\}\{\sa_{\alpha\dot\alpha},\sa_{\beta\dot\beta}\}+\frac{1}{2}\{\sa^{\alpha\dot\alpha},\phi^{\ib}\}\{\sa_{\alpha\dot\alpha},\phi_{\ib}\}+\frac{1}{4}\{\phi^{\ib},\phi^{\jb}\}\{\phi_{\ib},\phi_{\jb}\}+\ldots\Big)\nn\\
    &+S_{\BG}\,.
\end{align}
Here, the integration is over twistor space with the symplectic measure 
\begin{align}\label{symplectic-measure-PT}
   \mho= \tK\, \rho_{\cM}\,d^4y\,,\qquad \tK:=\frac{\langle \lambda\,d\lambda\rangle \wedge \langle \bar\mu\,d\bar\mu\rangle }{\langle \lambda\,\bar\mu\rangle^2}\,,\qquad \rho_M:=\frac{1}{\sinh(\tau)}
\end{align}
with $\tK$ being the top form on $\P^1$ fiber space having the standard normalization\footnote{This is nothing but the first Chern class of $T\P^1$.}
\begin{align}
\label{K-normalization}
    \int_{\P^1}\tK=-2\pi \, .
\end{align}
Note that the term $\{\ttb_{\mu},\ttb_{\nu}\}\{\sa^{\mu},\sa^{\nu}\}$ is suppressed in the flat limit of the FLRW matrix model-like spacetime, where $R\rightarrow 
\infty$ and $\sinh(\tau)$ can be treated as a (large) constant at late time $\tau$ as explained in \cite{Steinacker:2023zrb}. Therefore, after adding a gauge fixing term $-\frac{1}{2}\{\ttb_{\alpha\dot\alpha},\sa^{\alpha\dot\alpha}\}^2$, the spacetime kinetic action will come solely from the term $\frac 12 \{\ttb^{\alpha\dot\alpha},\sa^{\beta\dot\beta}\}\{\ttb_{\alpha\dot\alpha},\sa_{\beta\dot\beta}\}$.

\paragraph{Spacetime action.} For concreteness as well as simplicity, we will only derive the spacetime action of the Yang-Mills sector in the Lorentzian \hsikkt in the flat limit. All technical issues are included in Appendix \ref{app:A}. 

$\symknight$ \underline{\emph{Kinetic action.}} Using the effective metric \eqref{effectivemetricFLRW}, we can write the gauge-fixed kinetic term explicitly as
\begin{align}
    S_2\approx \frac{1}{2}\int_{\P^{1,2}}d^4y\,\tK\, \sinh(\tau) \eta^{\mu\nu}\p_{\mu}\sa^{\rho}\p_{\nu}\sa_{\rho}\,,
\end{align}
which is equivalent to
\begin{align}
    S_2\approx -\int_{\P^{1,2}} \d^4y\,\tK\,\sinh(\tau)\sa^{\alpha\dot\alpha}\Box\sa_{\alpha\dot\alpha}\,,\qquad \Box:=\frac{1}{2}\p_{\alpha\dot\alpha}\p^{\alpha\dot\alpha}\,.
\end{align}
Expanding $\sa_{\alpha\dot\alpha}=\sum_s\lambda^{\beta(s)}\bar\mu^{\beta(s)}\Big[A_{\beta(2s)\alpha\,\dot\alpha}+4s\,\eps_{\alpha\beta}\sA_{\beta(2s-1)\,\dot\alpha}\Big]$ and repeating the computation of the Euclidean \hsikkt in \cite{Steinacker:2023zrb} with some appropriate rescaling, we get 
\begin{equation}\label{S2-Apm}
    \begin{split}
    S_2&\approx+\pi  \sum_{s\geq 1}\frac{1}{s+1}\int d^4 y \sinh(\tau)\,\Big(A^{\alpha(2s-1)\,\dot\alpha}\Box A_{\alpha(2s-1)\,\dot\alpha}+\sA^{\alpha(2s-1)\,\dot\alpha}\Box \sA_{\alpha(2s-1)\,\dot\alpha}\Big)\\
    &\approx+\pi \sum_{s\geq 1}\frac{1}{s+1}\int d^4y \,\sinh(\tau)\A_-^{\alpha(2s-1)\,\dot\alpha}\Box\A^+_{\alpha(2s-1)\,\dot\alpha}\,.
    \end{split}
\end{equation}
where $ \A^{\alpha(2s-1)\,\dot\alpha}_{\pm}:=A^{\alpha(2s-1)\,\dot\alpha}\pm\im\,\sA^{\alpha(2s-1)\,\dot\alpha}$. We may interpret the higher-spin modes $\A_{\pm}$ as fields with positive/negative ``helicity'' $\pm s$, which carry $2(2s+1)$ dof. To remove the $\sinh(\tau)$ pre-factor in \eqref{S2-Apm}, we shall normalize each (higher-spin) field with a factor of $\sinh(\tau)^{-1/2}$. This normalization will effectively suppress the cubic action with a factor of $e^{-\frac{3}{2}\tau}$ and the quartic action with a factor of $e^{-3\tau}$ at late time regime, as shown in \cite{Steinacker:2023cuf}. 
With this consideration, the kinetic action reads
\begin{align}
    S_2\approx \sum_{s\geq 1}\frac{1}{s+1}\int d^4y\,\A_-^{\alpha(2s-1)\,\dot\alpha}\Box \A^+_{\alpha(2s-1)\,\dot\alpha}\,.
\end{align}
In terms of polarization tensors, we can write
\begin{align}
    \A^{\pm}_{\alpha(2s-1)\,\dot\alpha}=\eps^{\pm}_{\alpha(2s-1)\,\dot\alpha}\,e^{\im\, p\cdot y}\,,
\end{align}
in the local $4d$ regime. We will classify $\eps^{\pm}$ in the next section.

$\symbishop$ \underline{\emph{Cubic action.}} The cubic action of the YM part arises from 
\begin{align}
    S_3=\int \mho\,\{\ttb^{\alpha\dot\alpha},\sa^{\beta\dot\beta}\}\{\sa_{\alpha\dot\alpha},\sa_{\beta\dot\beta}\}\,.
\end{align}
Upon integrating out the fiber coordinates, we obtain
\begin{align}\label{cubic-action}
     S_3\approx  -\im L_{\rm NC}^2e^{-\frac{3}{2}\tau}\sum_{s_i}\int d^4y\,V_{\boldsymbol{m}}^{(s_1,s_2,s_3)}\,,\qquad \boldsymbol{m}=0,1,2,3\,,
\end{align}
in the late time regime using the canonical normalization for spinors, cf. \eqref{canonical-normalization}. Here, $V_{\boldsymbol{m}}^{(s_1,s_2,s_3)}$ are cubic sub-vertices with the external spins $(s_1,s_2,s_3)$ of a general cubic vertex
\begin{align}
    \cV^{\Lambda-\boldsymbol{m}}_{\boldsymbol{m}}=\sum_{s_i}V^{(s_1,s_2,s_3)}_{\boldsymbol{m}}
\end{align}
with maximal total spin $\Lambda=s_1+s_2+s_3$. The bold index $\boldsymbol{m}$ indicates whether $V$'s are the \emph{descendants} of the cubic vertices $V_{\boldsymbol{0}}^{(s_1,s_2,s_3)}$, i.e. the vertices where the second physical modes $\sA$'s show up. For instance, given a cubic sub-vertex $V_{\boldsymbol{0}}^{(s_1,s_2,s_3)}$, its descendants 
\begin{align}
    \big\{V_{\boldsymbol{\textcolor{red!70!black!100}{1}}}^{(\textcolor{red!70!black!100}{s_1-1},s_2,s_3)},V_{\boldsymbol{\textcolor{red!70!black!100}{1}}}^{(s_1,\textcolor{red!70!black!100}{s_2-1},s_3)},V_{\boldsymbol{\textcolor{red!70!black!100}{1}}}^{(s_1,s_2,\textcolor{red!70!black!100}{s_3-1})},V_{\boldsymbol{\textcolor{red!70!black!100}{2}}}^{(\textcolor{red!70!black!100}{s_1-1},\textcolor{red!70!black!100}{s_2-1},s_3)},V_{\boldsymbol{\textcolor{red!70!black!100}{2}}}^{(\textcolor{red!70!black!100}{s_1-1},s_2,\textcolor{red!70!black!100}{s_3-1})},\ldots\big\}
\end{align}
can be obtained by taking the trace wrt. $\eps$'s symbols if $s_i$ are high enough. Furthermore, as observed in \cite{Steinacker:2023cuf}, all interacting vertices exist only for $\Lambda\in 2\N$ while they vanish if $\Lambda$ is odd. For example, at $\Lambda=3,5,\ldots$, it can be checked explicitly that $V_{\boldsymbol{m}}^{(s_1,s_2,s_3)}=0$. 

For low values of $\Lambda$, e.g. $\Lambda=4,6$, the explicit spacetime cubic vertices of Lorentzian \hsikkt can be derived. The results are listed below.

$\bullet$ \underline{At $\Lambda=4$}, we have
\begin{subequations}
    \begin{align}
    V^{(2,1,1)}_{\boldsymbol{0}}&=-\p_{\tm\dot\tm}A_{\ta(2)\tv\,\dot\tv}\p^{\ta}{}_{\dot\bullet}A^{\tm\dot\tm}\p^{\ta\dot\bullet}A^{\tv\dot\tv}-y^{\ta \dot\ta}y^{\ta\dot\tb}(\p_{\tm\dot\tm}A_{\ta(2)\tv\,\dot\tv})(\p_{\bullet\dot\ta}A^{\tm\dot\tm})(\p^{\bullet}{}_{\dot\tb}A^{\tv\dot\tv})\,,\\
    V^{(1,2,1)}_{\boldsymbol{0}}&=-\p_{\tm\dot\tm}A_{\tv\dot\tv}\p_{\ta\dot\bullet}A^{\ta(2)\tm\,\dot\tm}\p_{\ta}{}^{\dot\bullet}A^{\tv\dot\tv}-y^{\ta\dot\ta}y^{\ta\dot\tb}(\p_{\tm\dot\tm}A_{\tv\dot\tv})(\p_{\bullet\dot\ta}A_{\ta(2)}{}^{\tm\dot\tm})(\p^{\bullet}{}_{\dot\tb}A^{\tv\dot\tv})\,,\\
    V^{(1,1,2)}_{\boldsymbol{0}}&=-\p_{\tm\dot\tm}A_{\tv\dot\tv}\p_{\ta\dot\bullet}A^{\tm\dot\tm}\p_{\ta}{}^{\dot\bullet}A^{\ta(2)\tv\,\dot\tv}-y^{\ta\dot\ta}y^{\ta\dot\tb}(\p_{\tm\dot\tm}A_{\tv\dot\tv})(\p_{\bullet\dot\ta}A^{\tm\dot\tm})(\p^{\bullet}{}_{\dot\tb}A_{\ta(2)}{}^{\tv\dot\tv})\,,\\
    V_{\boldsymbol{\textcolor{red!70!black!100}{1}}}^{(\redone,1,1)}&=(\p_{\tm\dot\tm}\sA_{\tn\dot\tv})(\p^{\tn}{}_{\dot\bullet}A^{\tm\dot\tm})(\p_{\tv}{}^{\dot\bullet}A^{\tv\dot\tv})+(\p_{\tm\dot\tm}\sA_{\tto\dot\tv})(\p_{\tv\dot\bullet}A^{\tm\dot\tm})(\p^{\tto\dot\bullet}A^{\tv\dot\tv})\nn\\
     &\quad -(y_{\tv}{}^{\dot\ta}y^{\ta\dot\tb}+y^{\ta\dot\ta}y_{\tv}{}^{\dot\tb})(\p_{\tm\dot\tm}\sA_{\ta\dot\tv})(\p_{\bullet\dot\ta}A^{\tm\dot\tm})(\p^{\bullet}{}_{\dot\tb}A^{\tv\dot\tv}) \,,\\
    V_{\boldsymbol{\textcolor{red!70!black!100}{1}}}^{(1,\redone,1)}&=(\p^{\tto}{}_{\dot\bullet}A_{\tv\dot\tv})(\p_{\tn\dot\bullet}\sA^{\tn\dot\tm})(\p_{\tto}{}^{\dot\bullet}A^{\tv\dot\tv})+(\p^{\tn}{}_{\dot\tm}A_{\tv\dot\tv})(\p_{\tn\dot\bullet}\sA^{\tto\dot\tm})(\p_{\tto}{}^{\dot\bullet}A^{\tv\dot\tv})\nn\\
    &\quad -(y^{\tm\dot\ta}y^{\ta\dot\tb}+y^{\ta\dot\ta}y^{\tm\dot\tb})(\p_{\tm\dot\tm}A_{\tv\dot\tv})(\p_{\bullet\dot\ta}\sA_{\ta}{}^{\dot\tm})(\p^{\bullet}{}_{\dot\tb}A^{\tv\dot\tv})\,,\\
    V_{\boldsymbol{\textcolor{red!70!black!100}{1}}}^{(1,1,\redone)}&=(\p_{\tm\dot\tm}A^{\tn}{}_{\dot\tv})(\p_{\tn\dot\bullet}A^{\tm\dot\tm})(\p_{\tto}{}^{\dot\bullet}\sA^{\tto\dot\tv})+(\p_{\tm\dot\tm}A^{\tto}{}_{\dot\tv})(\p_{\tn\dot\bullet}A^{\tm\dot\tm})(\p_{\tto}{}^{\dot\bullet}\sA^{\tn\dot\tv})\nn\\
    &\quad -(y^{\tv\dot\ta}y^{\ta\dot\tb}+y^{\ta\dot\ta}y^{\tv\dot\tb})(\p_{\tm\dot\tm}A_{\tv\dot\tv})(\p_{\bullet\dot\ta}A^{\tm\dot\tm})(\p^{\bullet}{}_{\dot\tb}\sA_{\ta}{}^{\dot\tv})\,.
\end{align}
\end{subequations}
Here, we have changed the Greek letters to Roman letters for better readability.

The results for $\cV^{\Lambda=6}_{\boldsymbol{0}}$ can be found in Appendix \ref{app:A}. Other descendants from $\cV^{\Lambda=6}_{\boldsymbol{0}}$ are omitted in this work since they do not add new insight to the study of scattering amplitudes. Note that we can substitute
\begin{align}
    A_{\ta(2s)\tm\,\dot\tm}=\frac{1}{2}\big(\A^+_{(\ta(2s)\tm)\,\dot\tm}+\A^-_{(\ta(2s)\tm)\,\dot\tm}\big)\,,\quad \sA_{\ta(2s-1)\,\dot\tm}=-\frac{\im}{2} \big(\A^+_{\ta(2s-1)\,\dot\tm}-\A^-_{\ta(2s-1)\,\dot\tm}\big)\,,
\end{align}
to obtain ``diagonalized'' vertices which have all possible configuration of ``helicities'', i.e.
\begin{align}
    (+++),(-++),(+-+),(++-),(--+),(-+-),(+--),(---)\,.
\end{align}

\underline{\emph{Remark.}} Observe that the structure of vertices in the Lorentzian case is different from Euclidean \hsikkt. On the one hand, the tangential components $y^{\alpha\dot\alpha}$ of $y^a$, which are used to define background and fluctuations for Euclidean \hsikkt, are anti-symmetric. This simplifies the structure of cubic and quartic vertices in the Euclidean theory and gives it some sort of self-dual or chiral feature on the Euclidean background; cf. 
\cite{Steinacker:2023zrb}. On the other hand, in the Lorentzian case, $\ttb^{\alpha\dot\alpha}$ cf. \eqref{eq:Hopf-t} is neither symmetric nor anti-symmetric. As a result, the Lorentzian theory really behaves as a higher-spin generalization of $\cN=4$ SYM. 
\medskip

$\symrook$ \underline{\emph{Quartic action.}} The quartic term at late time reads
\begin{align}\label{quartic-action}
    S_4\approx -L_{\rm NC}^4e^{-3\tau}\sum_{s_i}\int d^4y\,V^{(s_1,s_2,s_3,s_4)}_{\boldsymbol{m}}\,,\qquad \boldsymbol{m}=0,1,2,3,4\,. 
\end{align}
Once again, $V^{(s_1,s_2,s_3,s_4)}_{\boldsymbol{m}}$ vanishes if $\Lambda$ is odd. In this work, we compute the simplest example of the quartic vertex at $\Lambda=4$ where
\begin{align}
    V^{(1,1,1,1)}_{\boldsymbol{0}}&=(\p_{\tm\dot\bullet}A^{\ta\dot\ta})(\p^{\tp \dot\bullet}A^{\tb\dot\tb})(\p_{\tp\dot\diamond}A_{\ta\dot\ta})(\p^{\tm\dot\diamond}A_{\tb\dot\tb})+(\p_{\tm\dot\bullet}A^{\ta\dot\ta})(\p^{\tp\dot\bullet}A^{\tb\dot\tb})(\p^{\tm}{}_{\dot\diamond}A_{\ta\dot\ta})(\p_{\tp}{}^{\dot\diamond}A_{\tb\dot\tb})\nn\\
    &+y_{\te_1}{}^{\dot\tp}y^{\te_2\dot\tq}\Big[\p_{\te_2\dot\bullet}A^{\ta\dot\ta}\p^{\te_1\dot\bullet}A^{\tb\dot\tb}\p_{\diamond\dot\tp}A_{\ta\dot\ta}\p^{\diamond}{}_{\dot\tq}A_{\tb\dot\tb}+\p^{\te_1}{}_{\dot\bullet}A^{\ta\dot\ta}\p_{\te_2}{}^{\dot\bullet}A^{\tb\dot\tb}\p_{\diamond\dot\tp}A_{\ta\dot\ta}\p^{\diamond}{}_{\dot\tq}A_{\tb\dot\tb}\Big]\nn\\
    &+y_{\tf_1}{}^{\dot\tm}y^{\tf_2\dot\tn}\Big[\p_{\diamond\dot\tm}A^{\ta\dot\ta}\p^{\diamond}{}_{\dot\tn}A^{\tb\dot\tb}\p_{\tf_2\dot\bullet}A_{\ta\dot\ta}\p^{\tf_1\dot\bullet}A_{\tb\dot\tb}+\p_{\diamond\dot\tm}A^{\ta\dot\ta}\p^{\diamond}{}_{\dot\tn}A^{\tb\dot\tb}\p^{\tf_1}{}_{\dot\bullet}A_{\ta\dot\ta}\p_{\tf_2}{}^{\dot\bullet}A_{\tb\dot\tb}\Big]\nn\\
    &+\frac{1}{2}\big(y_{\tf_1}{}^{\dot\tm}y^{\tf_2\dot\tn}+y^{\tf_2\dot\tm}y_{\tf_1}{}^{\dot\tn}\big)\big(y_{\tf_2}{}^{\dot\tp}y^{\tf_1\dot\tq}+y^{\tf_1\dot\tp}y_{\tf_2}{}^{\dot\tq}\big)\p_{\diamond\dot\tm}A^{\ta\dot\ta}\p^{\diamond}{}_{\dot\tn}A^{\tb\dot\tb}\p_{\circ\dot\tp}A_{\ta\dot\ta}\p^{\circ}{}_{\dot\tq}A_{\tb\dot\tb}\,.
\end{align}
Note that all possible helicities are allowed. This is in contrast with the usual story of massless (quasi-) chiral higher-spin theories, where some of the configurations of the external helicities are absent by construction. The reason for this discrepancy is that most of the higher-spin modes in the present Lorentzian \hsikkt theory are would-be-massive modes and the interaction vertices are mildly Lorentz violated, at least in the present unitary or space-like formulation. 

\section{Yang-Mills sector tree-level scattering amplitudes}\label{sec:amplitudes}
This section studies some simple tree-level scattering amplitudes of the Lorentzian \hsikkt theory. 
We observe that the higher-spin tree-level amplitudes are non-trivial on-shell due to Lorentz-violating vertices, as in \cite{Steinacker:2023cuf}. Finally, we comment on the massless sector, where some simplifications occur.

\subsection{Kinematics and propagators}
In the following computations, we shall use the diagonalized higher-spin modes $\A^{\pm}$. We call these modes the ``diagonalized'' or ``helicity'' higher-spin modes. 

\paragraph{Kinematics.} Assuming $p_i^2=0$ in the locally flat and late time regime, we can write the massless 4-momentum $p^{\alpha\dot\alpha}_i$ of an external field $\A^{\pm}_i$ in spinorial form as $p_i^{\alpha\dot\alpha}=\rho_i^{\alpha}\tilde{\rho}_i^{\dot\alpha}$. 
There are two specific sectors that we will focus on in Lorentzian \hsikkt theory:

$\bullet$ \underline{\emph{Would-be-massive sector.}} 
Assuming factorization of spinors, i.e.
\begin{align}
    \zeta_{\alpha(2s)}=\zeta_{\alpha_1}\ldots\zeta_{\alpha_s}
\end{align}
etc., we propose the following ansatz for the positive and negative helicity polarization tensors associated with the would-be-massive external higher-spin states:
\begin{align}\label{YMhel}
    \eps^{i+}_{\alpha(2s-1)\,\dot\alpha}=\frac{\zeta_{\alpha(2s-1)}\tilde\kappa^i_{\dot\alpha}}{\langle\kappa_i\,\zeta\rangle^{2s-1}}\,,\qquad  \eps^{i-}_{\alpha(2s-1)\,\dot\alpha}&=\frac{\kappa^i_{\alpha(2s-1)}\,\tilde{\zeta}_{\dot\alpha}}{[\kappa_i\,\zeta]}\,,
\end{align}
where $(\zeta,\tilde\zeta)$ are auxiliary spinors and $(\kappa,\tilde \kappa)$ may be referred to as principal spinors. A priori, $ \kappa_i^{\alpha}\rho_{i\alpha}\neq 0$ and $\tilde\kappa_i^{\dot\alpha}\tilde\rho_{i\dot\alpha}\neq 0$ due to the fact that higher-spin fields in Lorentzian \hsikkt theory carry more dof. compared to conventional higher-spin theories. Nevertheless, the gauge-fixing condition \eqref{gaugefix-intertwiner} on twistor space descends to
\begin{align}\label{eq:divergence-free}
    p^{\alpha\dot\alpha}\eps^{\pm}_{\beta(2s-2)\alpha\,\dot\alpha}=0\,
\end{align}
on spacetime, which should hold for all physical modes. This implies
\begin{align}
    \zeta_\alpha p^{\alpha\dot\alpha}\tilde\kappa_{\dot\alpha} = 0=\kappa_{\alpha}p^{\alpha\dot\alpha}\tilde\zeta_{\dot\alpha}\,,
\end{align}
or equivalently $\rho_i^{\alpha}\simeq \kappa_i^{\alpha}$ and $\tilde\rho_i^{\dot\alpha}\simeq \tilde\kappa_i^{\dot\alpha}$ since $(\zeta,\tilde\zeta)$ are auxiliary spinors. As a result, the polarization tensors \eqref{YMhel} reduce to 
\begin{align}\label{YMhel-massless}
    \eps^{i+}_{\alpha(2s-1)\,\dot\alpha}=\frac{\zeta_{\alpha(2s-1)}\tilde\rho^i_{\dot\alpha}}{\langle i\,\zeta\rangle^{2s-1}}\,,\qquad  \eps^{i-}_{\alpha(2s-1)\,\dot\alpha}&=\frac{\rho^i_{\alpha(2s-1)}\,\tilde{\zeta}_{\dot\alpha}}{[i\,\tilde{\zeta}]}\,.
\end{align}
It is easy to check that the above representatives of the polarization tensors obey the normalization 
\be
\eps^{+}_{\alpha(2s-1)\,\dot\alpha}\eps_{-}^{\alpha(2s-1)\,\dot\alpha}=-1\,.
\ee
Furthermore, for the ansatz \eqref{YMhel-massless}, we can obtain the standard eom. for free positive/negative helicity massless higher-spin fields as \cite{Adamo:2022lah}:
\begin{subequations}\label{eq:nice-properties}
\begin{align}
 \partial_{\alpha}{}^{\dot\gamma}\A^{+}_{\alpha(2s-1)\,\dot\gamma}&=0\label{ph1}\,,\\
 \partial^{\beta\dot\alpha}\partial_{(\beta}{}^{\dot\gamma}\A^{-}_{\alpha(2s-1))\,\dot\gamma}&=0\label{nh1}\,.
\end{align}
\end{subequations}
This can be verified directly by plugging in the ansatz \eqref{YMhel-massless}.
Note, however, that after imposing \eqref{eq:divergence-free}, the diagonalized higher-spin fields $\A^{\pm}_{\alpha(2s-1)\,\dot\alpha}$ have
\begin{align}
    4s-(2s-1)=2s+1
\end{align}
degrees of freedom, which correspond to the would-be-massive modes. 
We thus recover the correct dof. for a massive higher-spin field in four dimension. 

$\bullet$ \underline{\emph{Massless sector.}} To reduce the degrees of freedom for the massless fields $\A^{\pm}$ to two, we can impose in addition the space-like or unitary gauge conditions:
\begin{align}\label{space-like-condition}
    y^{\alpha\dot\alpha}\eps^{\pm}_{\beta(2s-2)\alpha\,\dot\alpha}=0\,,
\end{align}
on the polarization tensors $\eps_{\alpha(2s_i-1)\,\dot\alpha}^{\pm}$. This space-like gauge conditions reduce further $(2s-1)$ dof. from $\A^{\pm}$, leaving us with $2$ dof of a massless higher-spin gauge field in four dimensions. Thus, after imposing \eqref{space-like-condition}, $\A^+$ and $\A^-$ each carries 1 dof. corresponding to two helicities of a massless higher-spin field. 

It is crucial to note that in this massless sector, the auxiliary spinors $(\zeta,\tilde\zeta)$ should drop out in any scattering amplitude expressions in the massless sector, as in the standard formulation of massless gauge fields. Below, we show that this is indeed the case. In particular, we can simplify the space-like gauge conditions to:
\begin{align}\label{eq:space-like-condition-spinors}
   y^{\alpha\dot\alpha}\zeta_{\alpha}\tilde\rho_{i\dot\alpha}=0=y^{\alpha\dot\alpha}\rho_{i\alpha}\tilde\zeta_{\dot\alpha}\,.
\end{align}
This will help us eliminate some explicit Lorentz-violating structures of the 3-point scattering amplitudes in the massless sector.

\paragraph{Propagator.} Recall that the kinetic action of the Yang-Mills sector after rescaling is
\be\label{FreeGT}
S_2=\sum_s\frac{1}{s+1}\int\d^4y\,\A^{\alpha(2s-1)\,\dot\alpha}_-\Box \A^+_{\alpha(2s-1)\,\dot\alpha}\,,\qquad s\geq 1\,.
\ee
Here, the fields comprise all would-be massive as well as massless modes.
Thus, up to normalization, the propagator between positive and negative helicity fields reads
\begin{align}\label{propagator}
    \langle \A^{+}_{\alpha(2s-1)\,\dot\alpha}(p)\A_{-}^{\beta(2s'-1)\,\dot\beta}(p')\rangle \sim \delta^4(p+p')\delta_{s,s'}\frac{\delta_{(\alpha_1}{}^{(\beta_1}\ldots \delta_{\alpha_{2s-1})}{}^{\beta_{2s'-1})}\delta_{\dot\alpha}{}^{\dot\beta}}{p^2}\,,
\end{align}
where $\delta_{x,y}$'s denote the Kronecker deltas. It is natural to view the $\A_{\alpha(2s-1)\,\dot\alpha}$ as chiral fields \cite{Krasnov:2021nsq}, because they have more un-dotted than dotted spinorial indices. 


\subsection{Three-point tree-level amplitudes} 
Since the interactions of Lorentzian \hsikkt mildly violate local Lorentz invariance, the on-shell tree-level 3-point amplitudes are \emph{not} vanishing a priori. Here, we only compute some of the 3-point amplitudes associated with the sub-vertices $V^{(s_1,s_2,s_3)}_{\boldsymbol{0}}$ in Section \ref{sec:SO(1,3)spinors-IKKT}. In the following, we denote the $3$-point scattering amplitude as $\cM_3(1_{s_1}^{h_1},2_{s_2}^{h_2},3_{s_3}^{h_3})$, where $h_i=\pm$ is the associated helicity of the $i$th particle.

Let us begin with $(V^{(2,1,1)}_{\boldsymbol{0}},V^{(1,2,1)}_{\boldsymbol{0}},V^{(1,1,2)}_{\boldsymbol{0}})$. In diagonalized or helicity form, they read
\begin{subequations}
    \begin{align}
    V^{(2,1,1)}_{\boldsymbol{0}}&=-\frac{1}{8}\p_{\tm\dot\tm}\big[\A^+_{\ta(2)\tv\,\dot\tv}+\A^-_{\ta(2)\tv\,\dot\tv}\big]\p^{\ta}{}_{\dot\bullet}\big[\A_+^{\tm\dot\tm}+\A_-^{\tm\dot\tm}\big]\p^{\ta\dot\bullet}\big[\A_+^{\tv\dot\tv}+\A_-^{\tv\dot\tv}\big]\nn\\
    &\quad -\frac{1}{8}y^{\ta \dot\ta}y^{\ta\dot\tb}\p_{\tm\dot\tm}\big[\A^+_{\ta(2)\tv\,\dot\tv}+\A^-_{\ta(2)\tv\,\dot\tv}\big]\p_{\bullet\dot\ta}\big[\A_+^{\tm\dot\tm}+\A_-^{\tm\dot\tm}\big]\p^{\bullet}{}_{\dot\tb}\big[\A_+^{\tv\dot\tv}+\A_-^{\tv\dot\tv}\big]\,,\\
    V^{(1,2,1)}_{\boldsymbol{0}}&=-\frac{1}{8}\p_{\tm\dot\tm}\big[\A^+_{\tv\dot\tv}+\A^-_{\tv\dot\tv}\big]\p_{\ta\dot\bullet}\big[\A_+^{\ta(2)\tm\,\dot\tm}+\A_-^{\ta(2)\tm\,\dot\tm}\big]\p_{\ta}{}^{\dot\bullet}\big[\A_+^{\tv\dot\tv}+\A_-^{\tv\dot\tv}\big]\nn\\
    &\quad-\frac{1}{8}y^{\ta\dot\ta}y^{\ta\dot\tb}\p^{\tm\dot\tm}\big[\A^+_{\tv\dot\tv}+\A^-_{\tv\dot\tv}\big]\p_{\bullet\dot\ta}\big[\A^+_{\ta(2)\tm\,\dot\tm}+\A^-_{\ta(2)\tm\,\dot\tm}\big]\p^{\bullet}{}_{\dot\tb}\big[\A_+^{\tv\dot\tv}+\A_-^{\tv\dot\tv}\big]\,,\\
    V^{(1,1,2)}_{\boldsymbol{0}}&=-\frac{1}{8}\p_{\tm\dot\tm}\big[\A^+_{\tv\dot\tv}+\A^-_{\tv\dot\tv}\big]\p_{\ta\dot\bullet}\big[\A_+^{\tm\dot\tm}+\A_-^{\tm\dot\tm}\big]\p_{\ta}{}^{\dot\bullet}\big[\A_+^{\ta(2)\tv\,\dot\tv}+\A_-^{\ta(2)\tv\,\dot\tv}\big]\nn\\
    &\quad -\frac{1}{8}y^{\ta\dot\ta}y^{\ta\dot\tb}\p_{\tm\dot\tm}\big[\A_+^{\tv\dot\tv}+\A_-^{\tv\dot\tv}\big]\p_{\bullet\dot\ta}\big[\A_+^{\tm\dot\tm}+\A_-^{\tm\dot\tm}\big]\p^{\bullet}{}_{\dot\tb}\big[\A^+_{\ta(2)\tv\,\dot\tv}+\A^-_{\ta(2)\tv\,\dot\tv}\big]\,.
\end{align}
\end{subequations}
The amplitudes of the above are:
\begin{subequations}
    \begin{align}
        \cM_3(1^{\pm}_2,2^{\pm}_1,3^{\pm}_1)&\sim +e^{-\frac{3}{2}\tau}L_{\rm NC}^2\,[2\,3]( \eps_2^{\pm,\tm\,\dot\tm}p_{1,\tm\dot\tm})\big(\eps^{\pm,\tv\dot\tv}_3\rho_2^{\ta}\rho_3^{\ta}\eps^{\pm}_{1,\ta(2)\tv\,\dot\tv}\big)\nn\\
        &\quad+ e^{-\frac{3}{2}\tau}L_{\rm NC}^2\langle 2\,3\rangle(\eps_2^{\pm,\tm\dot\tm}p_{1,\tm\dot\tm})(\eps_3^{\pm,\tv\dot\tv}y^{\ta\dot\ta}y^{\ta\dot\tb}\tilde \rho_{2\dot\ta}\tilde\rho_{3\dot\tb}\eps^{\pm}_{1,\ta(2)\tv\,\dot\tv})\nn\\
        &\quad+(\text{permutation})\,,\\
        \cM_3(1^{\pm}_1,2^{\pm}_2,3^{\pm}_1)&\sim +e^{-\frac{3}{2}\tau}L_{\rm NC}^2\,[2\,3]\big(\eps_2^{\pm,\ta(2)\tm\,\dot\tm}\rho_{2\ta}\rho_{3\ta}p_{1,\tm\dot\tm}\big)(\eps_{3\,\pm}^{\tv\dot\tv}\eps^{\pm}_{1\tv\,\dot\tv})\nn\\
        &\quad + e^{-\frac{3}{2}\tau}L_{\rm NC}^2\langle 2\,3\rangle (p_1^{\tm\dot\tm}\eps^{\pm}_{2,\ta(2)\tm\,\dot\tm} y^{\ta\dot\ta}y^{\ta\dot\tb}\tilde\rho_{2\dot\ta}\tilde\rho_{3\dot\tb})(\eps_{1,\tv\dot\tv}^{\pm}\eps_{3,\pm}^{\tv\dot\tv})\nn\\
        &\quad+(\text{permutation})\,,\\
        \cM_3(1^{\pm}_1,2^{\pm}_1,3^{\pm}_2)&\sim +e^{-\frac{3}{2}\tau}L_{\rm NC}^2\,[2\,3]( \eps_2^{\pm,\tm\,\dot\tm}p_{1,\tm\dot\tm})\big(\eps_{3\pm}^{\ta(2)\tv\,\dot\tv}\rho_{2\ta}\rho_{3\ta}\eps^{\pm}_{1,\tv\dot\tv}\big)\nn\\
        &\quad +e^{-\frac{3}{2}\tau}L_{\rm NC}^2\langle 2\,3\rangle(\eps_2^{\pm,\tm\dot\tm}p_{1,\tm\dot\tm})(\eps_{1,\pm}^{\tv\dot\tv}y^{\ta\dot\ta}y^{\ta\dot\tb}\eps^{\pm}_{3,\ta(2)\tv\,\dot\tv}\tilde\rho_{2\dot\ta}\tilde\rho_{3\dot\tb})\nn\\
        &\quad+(\text{permutation})\,.
    \end{align}
\end{subequations}
Note that the terms with explicit $y$-factors are invariant under the $SL(2,\C)$ isometry, but (mildly)
 Lorentz-violating. We show in the following that these Lorentz-violating contributions vanish in the massless sector, i.e. after imposing \eqref{eq:space-like-condition-spinors}.
\\
Assuming spinor factorization, we can use \eqref{YMhel-massless} to reduce the above to (see Appendix \ref{app:B})
\begin{subequations}
    \begin{align}
    \cM_3(1^{-}_2,2^{+}_1,3^{+}_1)&\sim e^{-\frac{3}{2}\tau}L_{\rm NC}^2\,[2\,3]\frac{\langle \zeta\,1\rangle^2 [2\,1][3\,\zeta]\langle 2\,1\rangle\langle3\,1\rangle}{\langle 2\,\zeta\rangle \langle 3\,\zeta\rangle [1\,\zeta]}\nn\\
    &\quad +e^{-\frac{3}{2}\tau}L_{\rm NC}^2\langle 2\,3\rangle\frac{\langle \zeta\,1\rangle^2[2\,1][3\,\zeta]}{\langle 2\,\zeta\rangle\langle3\,\zeta\rangle[1\,\zeta]}\big(y^{\ta\dot\ta}y^{\ta\dot\ta}\rho^1_{\ta(2)}\tilde\rho_{2\dot\ta}\tilde\rho_{3\dot\tb}\big)-(2\leftrightarrow 3)\,,\\
    \cM_3(1^{+}_2,2^{+}_1,3^{-}_1)&\sim e^{-\frac{3}{2}\tau}L_{\rm NC}^2\,[2\,3]\frac{[1\,2][\zeta\,1]\langle3\,\zeta\rangle^2}{\langle 1\,\zeta\rangle^2[3\,\zeta]}\nn\\
    &\quad +e^{-\frac{3}{2}\tau}L_{\rm NC}^2\langle 2\,3\rangle\frac{[1\,2]\langle3\,\zeta\rangle[\zeta\,1]}{\langle 2\,\zeta\rangle[3\,\zeta]\langle 1\,\zeta\rangle^2}\big(y^{\ta\dot\ta}y^{\ta\dot\tb}\tilde\rho_{2\dot\ta}\tilde\rho_{3\dot\tb}\zeta_{\ta(2)}\big) +(1\leftrightarrow 2)\,,\\
        \cM_3(1^{-}_2,2^{-}_1,3^{+}_1)&\sim -e^{-\frac{3}{2}\tau}L_{\rm NC}^2\,[2\,3]\frac{\langle 1\,2\rangle^2\langle\zeta\,1\rangle[3\,\zeta]\langle3\,1\rangle}{[2\,\zeta]\langle 3\,\zeta\rangle}\nn\\
        &\quad +e^{-\frac{3}{2}\tau}L_{\rm NC}^2\langle2\,3\rangle\frac{\langle 1\,2\rangle\langle\zeta\,1\rangle[3\,\zeta] }{[2\,\zeta]\langle 3\,\zeta\rangle}\big(y^{\ta\dot\ta}y^{\ta\dot\tb}\tilde \rho_{2\dot\ta}\tilde\rho_{3\dot\tb}\rho^1_{\ta(2)}\big)+(1\leftrightarrow 2)\,,\\
    \cM_3(1^{+}_2,2^{-}_1,3^{-}_1)&\sim e^{-\frac{3}{2}\tau}L_{\rm NC}^2\,[2\,3]\frac{[\zeta\,1]^2\langle 2\,1\rangle\langle 3\,\zeta\rangle^2\langle 2\,\zeta\rangle}{[2\,\zeta][3\,\zeta]\langle 1\,\zeta\rangle^3}\nn\\
    &\quad +e^{-\frac{3}{2}\tau}L_{\rm NC}^2\,\langle2\,3\rangle\frac{[\zeta\,1]^2\langle 2\,1\rangle\langle 3\,\zeta\rangle}{[2\,\zeta][3\,\zeta]\langle 1\,\zeta\rangle^3}\big(y^{\ta\dot\ta}y^{\ta\dot\tb}\tilde\rho_{2\dot\ta}\tilde\rho_{3\dot\tb}\zeta_{\ta(2)}\big)-(2\leftrightarrow 3)\,
    \end{align}
\end{subequations}
for $\cM_3(1^{\pm}_2,2^{\pm}_1,3^{\pm}_1)$. Observe that if $\A_1$ and $\A_3$ have the same helicity, the amplitudes $\cM_3(1^{\pm},2^{h_2},3^{\pm})$ vanish. All other 3-point amplitudes of $\cV^{\Lambda=4}_{\boldsymbol{0}}$ and their descendants can be computed analogously. We present some partial results in Appendix \ref{app:B}.

To cross the spin-2 barrier, we shall consider the example of $V^{(3,2,1)}_{\boldsymbol{0}}$ in its helicity form:
\begin{align}
    V_{\boldsymbol{0}}^{(3,2,1)}=\frac{1}{24}(\p_{\tm\dot\tm}\A^{\pm}_{\ta(2)\tto\tn\tv\,\dot\tv})(\p^{\tn}{}_{\dot\bullet}\A_{\pm}^{\ta(2)\tm\,\dot\tm})(\p^{\tto\dot\bullet}\A_{\pm}^{\tv\,\dot\tv})+\cC_{(3,2,1)}(y)
\end{align}
whose 3-point amplitudes read
\begin{align}
    \cM_3(1_3^{\pm},2_2^{\pm},1_1^{\pm})\sim e^{-\frac{3}{2}\tau}L_{\rm NC}^2\Big([2\,3]\big(\rho_2^{\tn}\rho_3^{\tto}\eps_{3\pm}^{\tv\,\dot\tv}\eps^{\pm}_{1,\ta(2)\tto\tn\tv\,\dot\tv}\eps_{2\pm}^{\ta(2)\tm\,\dot\tm}p_{1\tm\,\dot\tm}\big)+\cC_{(3,2,1)}(y)\Big)+(\text{permutation})\,,
\end{align}
where $\cC(y)$ are contributions that are $y$-dependent. Other cases can be done analogously. 

The sector of would-be massive modes is rather mysterious, since it has no counterpart in other more standard higher-spin theories.\footnote{See \cite{Metsaev:2005ar,Metsaev:2007rn,Metsaev:2022yvb} for classification of all Lorentz-invariant vertices in $4d$. It would be interesting to make a comparison between our results and what was found by Metsaev. See also e.g. \cite{Buchbinder:2022lsi,Skvortsov:2023jbn} for a covariant construction of cubic vertices involving massive higher-spin fields.}
However, we can show that 
 when restricting these amplitudes to the massless sector, one recovers gauge invariant amplitudes where all $(\zeta,\tilde\zeta)$ drop out of the final answers. 

$\bullet$ \underline{\emph{Massless sector.}} When we go to the massless sector, all explicit Lorentz-violation contributions, i.e. terms with $y$'s factors, can be simplified by imposing the space-like gauge conditions \eqref{space-like-condition} and \eqref{eq:space-like-condition-spinors}. 
As a result, the above $\cM_3(1^{\pm},2^{h_2},3^{\pm})$ amplitudes contain only contributions which are purely spinor-dependent. 
Using spinor-helicity techniques outlined in \cite{Adamo:2022lah,Tran:2022amg}, we get, for instance
\begin{subequations}
    \begin{align}
    \cM_3(1^{-}_2,2^{+}_1,3^{+}_1)&= 0\,,\\   \cM_3(1^{+}_2,2^{+}_1,3^{-}_1)&\sim e^{-\frac{3}{2}\tau}L_{\rm NC}^2\frac{[1\,2]^4}{[2\,3]^2}+(1\leftrightarrow 2)\,,\\
        \cM_3(1^{-}_2,2^{-}_1,3^{+}_1)&= 0\,,\\
    \cM_3(1^{+}_2,2^{-}_1,3^{-}_1)&= 0\,.
    \end{align}
\end{subequations}
Observe that some of the amplitudes are still non-trivial on-shell even in the massless case. This peculiar result reflects the Lorentz violation due to the Poisson structure \eqref{Poisson-spinor-1}, which involves the time-like vector field $\cY$ \eqref{timelikeT}. However, the cubic amplitudes for the lowest spin sector vanish on-shell, as they should.



Omitting massless 3-point amplitudes that are zero, we also find the following non-vanishing amplitudes
\begin{subequations}
    \begin{align}
        \cM_3(1_1^-,2_2^+,3_1^+)&\sim e^{-\frac{3}{2}\tau}L_{\rm NC}^2 \frac{[2\,3]^4}{[1\,3]^2}\,,\\
        \cM_3(1_1^+,2_2^+,3_1^-)&\sim e^{-\frac{3}{2}\tau}L_{\rm NC}^2 \frac{[1\,2]^3}{[2\,3]}\,,\\
        \cM_3(1_1^-,2_1^+,3_2^+)&\sim e^{-\frac{3}{2}\tau}L_{\rm NC}^2
        \frac{[2\,3]^4}{[1\,2]^2}\,,\\
        \cM_3(1_1^-,2_1^+,3_2^+)&\sim e^{-\frac{3}{2}\tau}L_{\rm NC}^2\frac{[2\,3]^2[3\,1]^2}{[1\,2]^2}\,.
    \end{align}
\end{subequations}
In the $\Lambda=6$ case, we compute only $\cM_3(1^{\pm}_3,2^{\pm}_2,1^{\pm}_1)$, and find that  $\cM_3(1^{\pm}_3,2^{\pm}_2,1^{\pm}_1)=0$ upon projecting to the massless sector. We expect that the non-vanishing massless amplitudes of $\cV^{\Lambda=6}_{\boldsymbol{0}}$  come from the sub-vertex $V^{(2,2,2)}_{\boldsymbol{0}}$. However, this is left as an exercise for the reader. 

\underline{\emph{Remark.}} The above results of massless scattering amplitudes indicate that the massless sector of the Yang-Mills term in Lorentzian \hsikkt is also chiral as in the case of Euclidean \hsikkt. This sector features a two-derivative chiral theory in the light-cone gauge. Indeed, let us project the above results to the light-cone gauge using the dictionary in \cite{Chalmers:1998jb,Bengtsson:2016jfk} where
\begin{align}
    i]=2^{1/4}\binom{\bar{\kbold}_i\,\beta_i^{-1/2}}{-\beta_i^{1/2}}\,,\qquad  i\rangle =2^{1/4}\binom{\kbold_i\,\beta_i^{-1/2}}{-\beta_i^{1/2}}\,,
\end{align}
where $\kbold_i^+\equiv \beta_i$, and $(\bar{\kbold}_i,\kbold_i)$ are referred to as \emph{transverse} derivatives. These data allow us to express the square and angle brackets as
\begin{align}
    [i\,j]=\sqrt{\frac{2}{\beta_i\beta_j}}\PPb_{ij}\,,\qquad \langle i\,j\rangle=\sqrt{\frac{2}{\beta_i\beta_j}}\PP_{ij}\,
\end{align}
for $\PPb_{ij}=\bar \kbold_i\beta_j-\bar{\kbold}_j\beta_i$ and $\PP_{ij}= \kbold_i\beta_j- \kbold_j\beta_i$. Using momentum conservation, one can show that
\begin{align} \label{PP1}
   \PPb_{12}=\PPb_{23}=\PPb_{31}= \PPb=\frac13\left[ (\beta_1-\beta_2)\bar{\kbold}_3+(\beta_2-\beta_3)\bar{\kbold}_1+(\beta_3-\beta_1)\bar{\kbold}_2\right]\,
\end{align}
at the level of 3-point amplitudes. Thus, in terms of these new variables, all massless 3-point amplitudes of Lorentzian \hsikkt scale as $\PPb^2$. Since there are no 3-point $\P$ vertices, the massless sector of the Lorentzian \hsikkt theory may be considered as a chiral theory in the sense of \cite{Metsaev:1991mt,Metsaev:1991nb}.\footnote{See also related discussions in \cite{Monteiro:2022xwq}.}
 However, it is important noting that these amplitudes violate local Lorentz invariance.


\paragraph{Comment on $n$-point amplitudes.} Computing higher-point amplitudes in the Yang-Mills sector is a non-trivial task. In particular, while some exchange channels vanish at the level of vertices $\cV^{\Lambda}_{\boldsymbol{0}}$ with maximal total spin $\Lambda\geq 4$, some descendants vertices do participate in higher-spin scatterings. Consider, for instance, the simplest example $\cM_4(1^{\pm},1^{\pm},1^{\pm},1^{\pm})$ of 4-point tree-level amplitude. As shown in Appendix \ref{app:A}, while vertices such ash $V_{\boldsymbol{0}}^{(1,1,1)}=0$ and $V_{\boldsymbol{0}}^{(s\geq 2,1,1)}=0$ vanish, the descendants sub-vertices such as $(V^{(\redone,1,1)}_{\boldsymbol{\redone}},V^{(1,\redone,1)}_{\boldsymbol{\redone}},V^{(1,1,\redone)}_{\boldsymbol{\redone}})$ are non-trivial. 
Nevertheless, this may be simplified if we consider only the massless sector, where most of the vertices vanish on-shell.

Finally, since the cubic and quartic vertices are suppressed by $e^{-\frac{3}{2}\tau}$ and $e^{-3\tau}$ factors, all $n$-point scattering amplitude are suppressed by a factor of $e^{-\frac{3}{2}(n-2)\tau}$ in the late time regime.


\section{Discussion}\label{sec:discussion}

In this work, we establish a global spinorial description for Lorentzian \hsikkt theory using Weyl spinors, which transform in the fundamental of the non-compact space-like isometry group $SL(2,\C)\simeq SO(1,3)$. We obtain the spacetime action of the $\hs$-valued Yang-Mills sector in \hsikkt theory in spinorial form, and study some tree-level 3-point higher-spin scattering amplitudes. Remarkably, they can be non-trivial on-shell due to mild Lorentz violation, but are strongly suppressed by a time-dependent coupling $e^{-\frac{3}{2}\tau}$ in the late-time and flat regime.  These results are consistent with a related study of similar amplitudes in the vectorial formulation of the same model in \cite{Steinacker:2023cuf}.

The technical tools presented in this work should allow in principle to extract all vertices in the remaining sectors and study higher-point tree-level amplitudes, and (eventually) also loop amplitudes of the Lorentzian \hsikkt under consideration. However, a better organization for the spacetime vertices may be needed to simplify this task. In particular, this should allow to explore gravitational interactions between spin-2 fields in Lorentzian \hsikkt and to compare them with standard GR, see also \cite{Kawai:2016wfh} in a similar context as ours. However, this requires including one-loop amplitudes, since the Einstein-Hilbert action only arises at one loop. That task is postponed to future work.

Given that Lorentz \hsikkt can have non-trivial scatterings from the would-be-massive higher-spin modes, it would be interesting to 
see whether these modes can acquire masses from quantum corrections. Intriguingly, we observe that some 3-point amplitudes in the massless sector of the model do not vanish due to Lorentz-violating structures encoded in the Poisson structure $\theta^{\alpha\dot\alpha\beta\dot\beta}$ cf. \eqref{Poisson-theta-spinor}. Nevertheless, local Lorentz invariance is expected to be recovered from the $\hs$ gauge symmetry in this model, as observed in \cite{Steinacker:2023cuf}.


\section*{Acknowledgement}
This work is supported by the Austrian Science Fund (FWF) grant P36479.

\appendix

\section{Algebra of functions on \texorpdfstring{$H^4_N$}{H4N} and doubleton representations}\label{app:C}
In the present approach, the 3+1 dimensional spacetime brane $\cM^{1,3}$ of interest is a projection of the almost-commutative 4-hyperboloid $H^4_N$ to $\R^{1,3}$. The semi-classical $H_N^4$ is defined by the following $\mso(1,4)$-covariant relations \cite{Sperling:2018xrm}
\begin{subequations}
\label{eq:so(4,2)algebra}
\begin{align}
    \{m_{ab},m_{cd}\}&=+(m_{ad}\eta_{bc}-m_{ac}\eta_{bd}-m_{bd}\eta_{ac}+m_{bc}\eta_{ad})\,,\\
    \{m_{ab},y_c\}&=+(y_a\eta_{bc}-y_b\eta_{ac})\,,\\
    \{y_a,y_b\}&=\theta^{ab}=-\ell_p^2\, m_{ab}\,,\qquad \qquad \qquad \quad \ \  a,b=0,1,2,3,4\,,\\
    y_ay^a&=-y_0^2+y_{\ta}y^{\ta}=-R^2=-\frac{\ell_p^2N^2}{4}\,,\qquad \ta=1,2,3,4\,,\label{ysphere}
\end{align}
\end{subequations}
where $m_{ab}$ are generators of $\mso(1,4)$ and $\eta_{ab}=\diag(-,+,+,+,+)$ is the metric of $\R^{1,4}$. The self-duality condition
\begin{align}\label{eq:selfdualityso(42)}
    \epsilon_{abcde}m^{ab}y^{c}=-\frac{4N}{\ell_p}m_{de}\,
\end{align}
can be used to reduce
the algebra of functions on $H_N^4$ to the commutative algebra $\R[y^a,m^{ab}]$ of polynomial functions in $y^a$ and $m^{ab}$ corresponding to the Young diagrams
\begin{align}
    C^{\infty}(H^4_N)=\sum_{k,m}f_{c(k)a(m),b(m)}y^{c(k)}m^{ab}\ldots m^{ab}=\bigoplus_{k,m}\ \parbox{85pt}{{\bep(70,30)\unitlength=0.4mm%
\put(0,3){\RectBRowUp{7}{4}{$m+k$}{$m$}}%
\eep}}\,.
\end{align}
This can be interpreted in terms of higher spin ($\hs$)- valued functions on $H^4$, or as functions on 3-dimensional non-compact complex twistor space $\P^{1,2}$.  
Notice that we can view the first three relations in \eqref{eq:so(4,2)algebra} as the Lie algebra $\mso(2,4)$ with generators $m_{\tA\tB}$ obeying
\begin{align}\label{so(2,4)-Alg}
    \{m_{\tA\tB},m_{\tC\tD}\}=\Big(\eta_{\tA\tC}m_{\tB\tD}-\eta_{\tA\tD}m_{\tB\tC}-\eta_{\tB\tC}m_{\tA\tD}+\eta_{\tB\tD}m_{\tA\tC}\Big)\,, \quad \tA,\tB=0,1,\ldots,5\,.
\end{align}
where $\eta^{\tA\tB}=\diag(-,+,+,+,+,-)$. The complete $\mso(2,4)$-invariant relations can be also formulated in terms of symmetric $\mso(1,4)$ generators $l^{AB}$ as follows \cite{Steinacker:2023zrb}:
\begin{subequations}\label{eq:sp(4)relations}
\begin{align}
    \{l^{AB},l^{CD}\}&=+(l^{AC}C^{BD}+l^{AD}C^{BC}+l^{BD}C^{AC}+l^{BC}C^{AD})\,,\\
    \{l^{AB},y^{CD}\}&=+(y^{AC}C^{BD}+y^{BC}C^{AD}-y^{AD}C^{BC}-y^{BD}C^{AC})\,,\\
    \{y^{AB},y^{CD}\}&=-(l^{AC}C^{BD}-l^{AD}C^{BC}-l^{BC}C^{AD}+l^{BD}C^{AC})\,,\\
    y_{AB}y^{AB}&=  l_{AB}l^{AB} = -4R^2\,,\\
    \epsilon_{ABCD}y^{AB} &= y_{CD}\,,\qquad \qquad \qquad \qquad \qquad \qquad A=1,2,3,4\,, \label{self-duality-sp4} 
\end{align}
\end{subequations}
where $C^{AB}$ is an $\mso(1,4)$-invariant matrix 
\begin{align}
    C^{AB}=-C^{BA}=\diag(\epsilon^{\alpha\beta},\epsilon^{\dot\alpha\dot\beta})
    \label{C-matrix}\,,\qquad \epsilon^{01}=-\epsilon^{10}=1\,,\quad \epsilon^{\alpha\beta}=\epsilon_{\alpha\beta}\,
\end{align}
and $\epsilon^{\alpha\beta}$ is the  $\msp(2)$-invariant tensor.
Notice that we have used $y^{AB}=-y^{BA}=\ell_p^{-1}y^a\gamma_a^{AB}$ and $l^{AB}= l^{BA}=\frac{1}{2} m^{ab}\Sigma^{AB}_{ab}$, where $\Sigma_{ab}:=\frac{\im}{4}[\gamma_a,\gamma_b]$ and $\gamma_a$ are the gamma matrices of $\mso(1,4)$, to map the standard vectorial $\mso(2,4)$ algebra in \cite{Sperling:2018xrm} to \eqref{eq:sp(4)relations}. The space of functions on $H^4_N$ can now be cast in terms of polynomial functions in $y^{AB}$ and $l^{AB}$, i.e.
\begin{align}
    C^{\infty}(H^4_N)=\sum_{k,m}f_{A(k)B(2m)|C(k)}y^{AC}\ldots y^{AC}l^{BB}\ldots l^{BB}=\bigoplus_{k,m}\ \parbox{75pt}{{\bep(70,30)\unitlength=0.4mm%
\put(0,3){\RectBRowUp{7}{4}{$k+2m$}{$k$}}%
\eep}}\,
\end{align}
As always, $\msu(2,2)$ and/or $\msp(2)$ indices are raised and lowered by
\begin{align}
    V_{B}C^{AB}=U^{A}\,,\qquad V^{B}C_{BA}=U_{A}\,,\qquad 
    u^{\alpha}=u_{\beta}\epsilon^{\alpha\beta}\,,\qquad u_{\alpha}=u^{\beta}\epsilon_{\beta\alpha}\,.
\end{align}

\paragraph{Doubleton minimal representations.}


To define the non-commutative fuzzy $H^4_N$ and 
to exhibit the hermiticity and positivity properties of the underlying $SU(2,2)$ generators, it is useful to realize our Weyl spinors $(\lambda,\mu)$ and $(\bar\mu,\bar\lambda)$ cf. Section \ref{sec:3} in terms of bosonic creation- and annihilation generators, 
cf. \cite{Govil:2013uta}.
In particular, we shall define
\begin{align}
    Z^A=\binom{\lambda^{\alpha}}{\mu^{\dot\alpha}}=\frac{1}{\sqrt{2}}\binom{a_i-b^{\dagger}_i}{a_i+b^{\dagger}_i}\,,\qquad \bar Z_A=(\bar\mu_{\alpha}, \bar\lambda_{\dot\alpha})=\frac{1}{\sqrt{2}}(a^{\dagger}_i+b_i\,,a^{\dagger}_i-b_i)\,,
\end{align}
where $(a_i,a_j^{\dagger})$ and $(b_i,b_j^{\dagger})$, with $i=1,2$,  satisfy the oscillator relations
\begin{align}
    [a_i,a_j^{\dagger}]=\delta_{ij}\,,\quad [b_i,b_j^{\dagger}]=\delta_{ij}\,,\qquad \text{and zero otherwise}\,.
\end{align}
Notice that they coincide with the $SU(2)_L \times SU(2)_R$ oscillators discussed e.g. in \cite{Sperling:2018xrm}.
We obtain
\begin{subequations}
    \begin{align}\label{oscillators-1}
    \hat \cN&=\bar{Z}_AZ^A=N_a-N_b-2\,,\qquad N_a:=a^{\dagger}_ia_i\,,\qquad N_b:=b^{\dagger}_ib_i\,,\\
    D&=\frac{1}{2}\bar{Z}_A(\gamma_4)^A{}_BZ^B=-\frac{\im}{2}(a_ib_i-a_i^{\dagger}b_i^{\dagger}) \,,
\end{align}
\end{subequations}
where the chiral basis for $\gamma$-matrices is given in \eqref{chiralso(1,5)basis}.
The conformal Hamiltonian $E$ is defined by 
\begin{align}
    E:=\frac{1}{2}\bar{Z}_A(\gamma_0)^A{}_BZ^A=\frac{1}{2}(N_a+N_b+2)\, > 0 .
\end{align}
Thus, given a Fock space built on the vacuum defined by $a_i|0\rangle =0=b_i|0\rangle$, we can construct the lowest weight states 
\begin{subequations}
    \begin{align}
        |\Omega\rangle_{a^{\dagger}}:&=\Big|1+\frac{s}{2},\frac{s}{2},0\Big\rangle = a_{i(s)}^{\dagger}|0\rangle \,,\\
        |\Omega\rangle_{b^{\dagger}}:&=\Big|1+\frac{s}{2},0,\frac{s}{2}\Big\rangle=b_{i(s)}^{\dagger}|0\rangle\,,
    \end{align}
\end{subequations}
which satisfy $\cL^-|\Omega\rangle = 0 = a_i b_j |\Omega\rangle$.
Acting on these states with all $\msu(2,2)$ generators, we obtain the \emph{doubleton minimal representations}, which form a discrete series of unitary representations. Since $\hat\cN$ commutes with $SU(2,2)$, these representations are characterized by the quantum number 
\begin{align}
    \hat \cN = N \ .
    \label{constraint-N}
\end{align}

\section{Integrating out fiber coordinates}\label{app:A}
This appendix provides a description for projecting/integrating out fiber coordinates $(\lambda,\bar\mu)$ of the $S^2\simeq \P^1$ fiber space over $H^4$. Note that $\bar\mu_{\alpha}$ is \emph{not} the quarternionic conjugation of $\lambda^{\alpha}$ since we are working with the non-compact space-like isometry subgroup $SL(2,\C)$ of the structure 
group $SU(2,2)$. Therefore, the standard integral over fibers in the compact case cf. \cite{Penrose:1985bww} must be modified accordingly. In particular, we shall prove the following Lemma. 

\begin{lemma}\label{lem:average} Let $(\lambda^{\alpha},\bar\mu_{\alpha})$ be the fiber coordinates of $\P^1$. We have
\begin{align}\label{eq:bridge}
    \int_{\P^1}\tK\,\frac{\bar\mu_{\alpha(s)}\lambda^{\beta(s)}}{\langle \lambda\,\bar\mu\rangle^s}=-\frac{2\pi}{s+1}\delta_{(\alpha_1}{}^{\beta_1}\ldots\delta_{\alpha_s)}{}^{\beta_s}\,.
\end{align}   
\end{lemma}

\begin{proof}
The tensor structure on the rhs follows from $SL(2,\C)$ invariance and symmetry, and the normalization can be checked by 
contracting the indices.
\end{proof}

Note that \eqref{eq:bridge} is unchanged if we use the normalization \eqref{canonical-normalization} in the main text. Now, the explicit spacetime action of the Yang-Mills sector in the Lorentzian case can be obtained by averaging over the $\P^1$ fiber using Lemma \ref{lem:average}.

\subsection{Kinetic action}
Recall that the kinetic action in the local $4d$ regime is
\begin{align}\label{K2generation}
    S_2=-\frac{1}{2}\sinh(\tau)\int \d^4y\,\int \tK\, \sa^{\alpha
    \,\dot\alpha}\Box \sa_{\alpha\,\dot\alpha}\,.
\end{align}
Expanding $\sa_{\alpha\dot\alpha}=\sum_s\lambda^{\beta(s)}\bar\mu^{\beta(s)}\Big[A_{\beta(2s)\alpha\,\dot\alpha}+4s\,\eps_{\alpha\beta}\sA_{\beta(2s-1)\,\dot\alpha}\Big]$, and rescaling 
\begin{align}
   A_{\alpha(2s+1)\,\dot\alpha}\mapsto \frac{1}{\langle \lambda\,\bar\mu\rangle^s}A_{\alpha(2s+1)\,\dot\alpha}\,,\qquad  \sA_{\alpha(2s-1)\,\dot\alpha}\mapsto \frac{1}{\langle \lambda\,\bar\mu\rangle^s}\sA_{\alpha(2s-1)\,\dot\alpha}
\end{align}
we get, for instance, the kinetic action for spin-1 fields as
    \begin{align}
        S_2^{(1)}=-\frac{\sinh(\tau)}{2}\int A^{\alpha\dot\alpha}\Box A_{\alpha\dot\alpha}-\frac{1}{\langle \lambda\,\bar\mu\rangle}\sA^{\beta\dot\alpha}\Box A^{\beta}{}_{\dot\alpha}\lambda_{\beta}\bar\mu_{\beta}+\frac{1}{\langle \lambda\,\bar\mu\rangle}A_{\beta}{}^{\dot\alpha}\Box \sA_{\beta\dot\alpha}\lambda^{\beta}\bar\mu^{\beta}+\sA^{\alpha\dot\alpha}\Box\sA_{\alpha\dot\alpha}\,.
    \end{align}
Using Lemma \ref{lem:average} and noting that both mixed terms vanish upon symmetrization, we obtain
    \begin{align}
        S_2^{(1)}=+\pi  \int \sinh(\tau) \,\Big(A^{\alpha\dot\alpha}\Box A_{\alpha\dot\alpha}+\sA^{\alpha\dot\alpha}\Box\sA_{\alpha\dot\alpha}\Big)\,.
    \end{align}
 Inductively, we get the full kinetic term as \cite{Steinacker:2023zrb}
\begin{equation}
    \begin{split}
    S_2&=+\pi   \sum_{s\geq 1}\frac{1}{s+1}\int \d^4 y \,\sinh(\tau)\Big(A^{\alpha(2s-1)\,\dot\alpha}\Box A_{\alpha(2s-1)\,\dot\alpha}+\sA^{\alpha(2s-1)\,\dot\alpha}\Box \sA_{\alpha(2s-1)\,\dot\alpha}\Big)\\
    &=\pi \sum_{s\geq 1}\frac{1}{s+1}\int \d^4 y \,\sinh(\tau)\A_-^{\alpha(2s-1)\,\dot\alpha}\Box\A^+_{\alpha(2s-1)\,\dot\alpha}\,.
    \end{split}
\end{equation}
where
\begin{align}
    \A^{\alpha(2s-1)\,\dot\alpha}_{\pm}:=A^{\alpha(2s-1)\,\dot\alpha}\pm\im\,\sA^{\alpha(2s-1)\,\dot\alpha}\,
\end{align}
may be interpreted as higher-spin fields with positive or 
negative ``helicity''. The factor $\pi$ can be absorbed by rescaling the action appropriately. Moreover, as discussed in \cite{Steinacker:2023cuf}, to have a canonical normalization, we will rescale 
\begin{align}
    \A^{\pm} \mapsto \frac{1}{\sqrt{\sinh(\tau)}}\A^{\pm}\,. 
\end{align}
The effect of this rescaling is that there will be no $\sinh(\tau)\approx e^{\tau}$ factor in the kinetic action in the late time regime. As a consequence, the cubic and quartic vertices will be suppressed by $e^{-\frac{3}{2}\tau}$ and $e^{-3\tau}$ factors, respectively. 

\subsection{Cubic action} 
The pre-higher-spin spacetime cubic vertex with the normalization \eqref{canonical-normalization} and the Poisson structures \eqref{eq:theta-normalized} reads
\begin{align}
    S_3&=\int \mho \,\{\ttb^{\mu},\sa^{\nu}\}\{\sa_{\mu},\sa_{\nu}\}\approx -\im L_{\rm NC}^2e^{-\frac{3}{2}\tau}\int d^4 y\int \tK  \,\p^{\mu}\sa^{\nu}\theta^{\alpha\dot\alpha\,\beta\dot\beta}\p_{\alpha\dot\alpha}\sa_{\mu}\p_{\beta\dot\beta}\sa_{\nu}\nn\\
    &\approx -\im L_{\rm NC}^2e^{-\frac{3}{2}\tau}\int  d^4y\int \tK\,\Big[\p^{\mu}\sa^{\nu}\big(\lambda^{\alpha}\bar\mu^{\beta}+\lambda^{\beta}\bar\mu^{\alpha})\p_{\alpha\dot\alpha}\sa_{\mu}\p_{\beta}{}^{\dot\alpha}\sa_{\nu}\nn\\
    &\qquad \qquad \qquad \qquad + \p^{\mu}\sa^{\nu}(y^{\gamma\dot\alpha}y^{\delta\dot\beta}+y^{\gamma\dot\beta}y^{\delta\dot\alpha})\lambda_{\gamma}\bar\mu_{\delta}\p_{\alpha\dot\alpha}\sa_{\mu}\p^{\alpha}{}_{\dot\beta}\sa_{\nu}\Big]\,,
\end{align}
where 
\begin{align}
    \sa_{\alpha\dot\alpha}&=\sum_s\frac{\lambda^{\beta(s)}\bar\mu^{\beta(s)}}{\langle \lambda\,\bar\mu\rangle^s}\Big[A_{\beta(2s)\alpha\,\dot\alpha}+4s\,\eps_{\alpha\beta} \sA_{\beta(2s-1)\,\dot\alpha}\Big]\,.
\end{align}
Note that we will rescale $\theta\mapsto \frac{1}{\langle \lambda\,\bar\mu\rangle}\theta$ for convenience. 

To improve our notation, we switch Greek letters back to Roman letters. For instance, $(\alpha,\dot\alpha) \mapsto (\ta,\dot\ta)$ etc. Then, we can write higher-spin cubic vertices as
\begin{align}
    S_3&=-\im L_{\rm NC}^2e^{-\frac{3}{2}\tau}\int d^4y\int_{\P1} \tK\,\Big(\circled{0}+\circled{1}+\circled{2}+\circled{3}\Big)
\end{align}
where 
\begin{subequations}
    \begin{align}
    \circled{0}&=\sum_{s_i}\digamma^{\ta(2s_1)|(\tn,\tto)}_{\tb(2s_2),\tc(2s_3)}(\p_{\tm\dot\tm}A_{\ta(2s_1)\tv\,\dot\tv})(\p_{\tn \dot\bullet}A^{\tb(2s_2)\tm\,\dot\tm})(\p_{\tto}{}^{\dot\bullet}A^{\tc(2s_3)\tv\,\dot\tv})\nn\\
    &\quad +\sum_{s_i}\widetilde{\digamma}^{\ta(2s_1)}_{\tb(2s_2),\tc(2s_3)|(\tn,\tto)}(y^{\tn\dot\ta}y^{\tto\dot\tb}+y^{\tto\dot\ta}y^{\tn\dot\tb})(\p_{\tm\dot\tm}A_{\ta(2s_1)\tv\,\dot\tv})(\p_{\bullet\dot\ta}A^{\tb(2s_2)\tm\,\dot\tm})(\p^{\bullet}{}_{\dot \tb}A^{\tc(2s_3)\tv\,\dot\tv})\,,\\
    \circled{1}&=\sum_{s_i}\digamma^{\ta(2s_1)|(\tn,\tto)}_{\tb(2s_2),\tc(2s_3)}\Big[-4s_3(\p_{\tm\dot\tm}A_{\ta(2s_1)\tc\,\dot\tv})(\p_{\tn\dot\bullet}A^{\tb(2s_2)\tm\,\dot\tm})(\p_{\tto}{}^{\dot\bullet}\sA^{\tc(2s_3-1)\,\dot\tv})\nn\\
    &\quad \qquad \qquad \qquad -4s_2(\p^{\tb}{}_{\dot\tm}A_{\ta(2s_1)\tv\,\dot\tv})(\p_{\tn\dot\bullet}\sA^{\tb(2s_2-1)\,\dot\tm})(\p_{\tto}{}^{\dot\bullet}A^{\tc(2s_3)\tv\,\dot\tv})\nn\\
    &\quad \qquad \qquad \qquad+4s_1(\p_{\tm\dot\tm}\sA_{\ta(2s_1-1)\,\dot\tv})(\p_{\tn\dot\bullet}A^{\tb(2s_2)\tm\,\dot\tm})(\p_{\tto}{}^{\dot\bullet}A_{\ta}{}^{\tc(2s_3)\,\dot\tv})\Big]\,,\\
    &\quad +\sum_{s_i}\widetilde\digamma^{\ta(2s_1)}_{\tb(2s_2),\tc(2s_3)|(\tn,\tto)}(y^{\tn\dot\ta}y^{\tto\dot\tb}+y^{\tto\dot\ta}y^{\tn\dot\tb})\Big[-4s_3(\p_{\tm\dot\tm}A_{\ta(2s_1)\tc\,\dot\tv})(\p_{\bullet\dot\ta}A^{\tb(2s_2)\tm\,\dot\tm})(\p^{\bullet}{}_{\dot\tb}\sA^{\tc(2s_3-1)\,\dot\tv})\nn\\
    &\quad \qquad \qquad \qquad -4s_2(\p^{\tb}{}_{\dot\tm}A_{\ta(2s_1)\tv\,\dot\tv})(\p_{\bullet\dot\ta}\sA^{\tb(2s_2-1)\,\dot\tm})(\p^{\bullet}{}_{\dot\tb}A^{\tc(2s_3)\tv\,\dot\tv})\nn\\
    &\quad \qquad \qquad \qquad+4s_1(\p_{\tm\dot\tm}\sA_{\ta(2s_1-1)\,\dot\tv})(\p_{\bullet\dot\ta}A^{\tb(2s_2)\tm\,\dot\tm})(\p^{\bullet}{}_{\dot\tb}A_{\ta}{}^{\tc(2s_3)\,\dot\tv})\Big]\,,\\
    \circled{2}&=\sum_{s_i}\digamma^{\ta(2s_1)|(\tn,\tto)}_{\tb(2s_2),\tc(2s_3)}\Big[+16s_2s_3(\p^{\tb}{}_{\dot\tm}A^{\tc}{}_{\ta(2s_1)\,\dot\tv})(\p_{\tn\dot\bullet}\sA^{\tb(2s_2-1)\,\dot\tm})(\p_{\tto}{}^{\dot\bullet}\sA^{\tc(2s_3-1)\,\dot\tv})\nn\\
    &\quad \qquad \qquad \qquad-16s_1s_3(\p_{\tm\dot\tm}\sA_{\ta(2s_1-1)\,\dot\tv})(\p_{\tn\dot\bullet}A^{\tb(2s_2)\tm\,\dot\tm})(\p_{\tto}{}^{\dot\bullet}\sA^{\tc(2s_3-1)\,\dot\tv})\eps_{\ta}{}^{\tc}\nn\\
    &\quad \qquad \qquad \qquad-16s_1s_2(\p^{\tb}{}_{\dot\tm}\sA_{\ta(2s_1-1)\,\dot\tv})(\p_{\tn\dot\bullet}\sA^{\tb(2s_2-1)\,\dot\tm})(\p_{\tto}{}^{\dot\bullet}A_{\ta}{}^{\tc(2s_3)\,\dot\tv})\Big]\nn\\
    &\quad +\sum_{s_i}\widetilde\digamma^{\ta(2s_1)}_{\tb(2s_2),\tc(2s_3)|(\tn,\tto)}(y^{\tn\dot\ta}y^{\tto\dot\tb}+y^{\tto\dot\ta}y^{\tn\dot\tb})\Big[+16s_2s_3(\p^{\tb}{}_{\dot\tm}A^{\tc}{}_{\ta(2s_1)\,\dot\tv})(\p_{\bullet\dot\ta}\sA^{\tb(2s_2-1)\,\dot\tm})(\p^{\bullet}{}_{\dot\tb}\sA^{\tc(2s_3-1)\,\dot\tv})\nn\\
    &\quad \qquad \qquad \qquad-16s_1s_3(\p_{\tm\dot\tm}\sA_{\ta(2s_1-1)\,\dot\tv})(\p_{\bullet\dot\ta}A^{\tb(2s_2)\tm\,\dot\tm})(\p^{\bullet}_{\dot\tb}\sA^{\tc(2s_3-1)\,\dot\tv})\eps_{\ta}{}^{\tc}\nn\\
    &\quad \qquad \qquad \qquad-16s_1s_2(\p^{\tb}{}_{\dot\tm}\sA_{\ta(2s_1-1)\,\dot\tv})(\p_{\bullet\dot\ta}\sA^{\tb(2s_2-1)\,\dot\tm})(\p^{\bullet}{}_{\dot\tb}A_{\ta}{}^{\tc(2s_3)\,\dot\tv})\Big]\,,\\
    \circled{3}&=\sum_{s_i}\digamma^{\ta(2s_1)|(\tn,\tto)}_{\tb(2s_2),\tc(2s_3)}64s_1s_2s_3\,\eps_{\ta}{}^{\tc}(\p^{\tb}{}_{\dot\tm}\sA_{\ta(2s_1-1)\,\dot\tv})(\p_{\tn\dot\bullet}\sA^{\tb(2s_2-1)\,\dot\tm})(\p_{\tto}{}^{\dot\bullet}\sA^{\tc(2s_3-1)\,\dot\tv})\nn\\
    &\quad +\sum_{s_i}\widetilde\digamma^{\ta(2s_1)}_{\tb(2s_2),\tc(2s_3)|(\tn,\tto)}64s_1s_2s_3\,(y^{\tn\dot\ta}y^{\tto\dot\tb}+y^{\tto\dot\ta}y^{\tn\dot\tb})\eps_{\ta}{}^{\tc}(\p^{\tb}{}_{\dot\tm}\sA_{\ta(2s_1-1)\,\dot\tv})(\p_{\bullet\dot\ta}\sA^{\tb(2s_2-1)\,\dot\tm})(\p^{\bullet}{}_{\dot\tb}\sA^{\tc(2s_3-1)\,\dot\tv})\,.
\end{align}
\end{subequations}
One may refer to
\begin{subequations}\label{fiber-factor}
    \begin{align}
    \digamma^{\ta(2s_1)|(\tn,\tto)}_{\tb(2s_2),\tc(2s_3)}:&=\frac{1}{\langle \lambda\,\bar\mu\rangle^{s_1+s_2+s_3+1}}\lambda^{\ta(s_1)}\bar\mu^{\ta(s_1)}(\lambda^{\tn}\bar\mu^{\tto}+\lambda^{\tto}\bar\mu^{\tn})\lambda_{\tb(s_2)}\bar\mu_{\tb(s_2)}\lambda_{\tc(s_3)}\bar\mu_{\tc(s_3)}\,,\\
    \widetilde\digamma^{\ta(2s_1)}_{\tb(2s_2),\tc(2s_3)|(\tn,\tto)}:&=\frac{1}{\langle \lambda\,\bar\mu\rangle^{s_1+s_2+s_3+1}}\lambda^{\ta(s_1)}\bar\mu^{\ta(s_1)}\lambda_{\tb(s_2)}\bar\mu_{\tb(s_2)}\lambda_{\tc(s_3)}\bar\mu_{\tc(s_3)}\lambda_{\tn}\bar\mu_{\tto}\,,
\end{align}
\end{subequations}
as the \emph{fiber factors}.

\underline{\emph{Normalization.}} Recall that Lemma \ref{lem:average} only depends on the numbers of $(\lambda,\bar\mu)$ upon averaging over fiber coordinates. Therefore, to obtain unique vertices $\cV_3^{(s_1,s_2,s_3)}$ with fixed external spins $(s_1,s_2,s_3)$ for the diagonalized higher-spin fields $\A$'s, one can adjust the spin entries of the fiber factors $\digamma^{\ta(2s_1)|(\tn,\tto)}_{\tb(2s_2),\tc(2s_3)}$ such that we will have equal numbers of $(\lambda,\bar\mu)$ upon averaging at the cubic. 
This 
leaves us with the time-dependent couplings $e^{-\frac{3\tau}{2}}$ in the local physical regime in the cubic sector.

\paragraph{Spacetime vertices.} For convenience, we denote the maximal total spin entering a cubic vertex as $\Lambda=s_1+s_2+s_3$ and the associated vertex as 
\begin{align}
    \cV_{\boldsymbol{m}}^{\Lambda-\boldsymbol{m}}=\sum_{s_i}V_{\boldsymbol{m}}^{(s_1,s_2,s_3)}\,,\qquad \boldsymbol{m}=0,1,2,3\,.
\end{align} 
we call $\cV_{\boldsymbol{m\geq 1}}$ the \emph{descendant vertices of $\cV_{\boldsymbol{0}}$}. This name is suggestive, since all descendant vertices associated with the second modes $\sA$'s can be extracted systematically by taking traces w.r.t. $\eps$ symbols. This will save us some effort of not performing the integral \eqref{eq:bridge} when the second higher-spin modes $\sA$ occur. Therefore, once we obtain all sub-vertices $V_{\boldsymbol{0}}$'s of $\cV^{\Lambda}_{\boldsymbol{0}}$, everything should follow.

$\bullet$ \underline{Vertices with maximal spin $\cV_{\boldsymbol{0}}^{\Lambda}$.} For the lowest level in spin, i.e. $\Lambda=3$ and $V_{\boldsymbol{0}}^{(1,1,1)}$, we can average over fiber coordinates and obtain
\begin{align}
    \cV_{\boldsymbol{0}}^3=V^{(1,1,1)}_{\boldsymbol{0}}&=\p_{\tm\dot\tm}A_{\tv\dot\tv}\p_{\tto\dot\bullet}A^{\tm\dot\tm}\p^{\tto\dot\bullet}A^{\tv\dot\tv}-\p_{\tm\dot\tm}A_{\tv\dot\tv}\p^{\tto}{}_{\dot\bullet}A^{\tm\dot\tm}\p_{\tto}{}^{\dot\bullet}A^{\tv\dot\tv}\nn\\
    &\quad +(y_{\tto}{}^{\dot\ta}y^{\tto\dot\tb}+y^{\tto\dot\ta}y_{\tto}{}^{\dot\tb})(\p_{\tm\dot\tm}A_{\tv\dot\tv})(p_{\bullet\dot\ta}A^{\tm\dot\tm})(\p^{\bullet}{}_{\dot\tb}A^{\tv\dot\tv})\nn\\
    &=0\,.
\end{align}
At $\Lambda=4$, we have three sub-vertices $V^{(2,1,1)}_{\boldsymbol{0}}$, $V^{(1,2,1)}_{\boldsymbol{0}}$ and $V^{(1,1,2)}_{\boldsymbol{0}}$. Their spacetime duals read
\begin{subequations}\label{V0-Lambda-4}
    \begin{align}
    V^{(2,1,1)}_{\boldsymbol{0}}&=-\frac{1}{2}\Big(\p_{\tm\dot\tm}A_{\tto\tn\tv\,\dot\tv}\p^{\tn}{}_{\dot\bullet}A^{\tm\dot\tm}\p^{\tto\dot\bullet}A^{\tv\dot\tv}+\p_{\tm\dot\tm}A_{\tn\tto\tv\,\dot\tv}\p^{\tn}{}_{\dot\bullet}A^{\tm\dot\tm}\p^{\tto\dot\bullet}A^{\tv\dot\tv}\Big)-y^{\ta \dot\ta}y^{\ta\dot\tb}(\p_{\tm\dot\tm}A_{\ta(2)\tv\,\dot\tv})(\p_{\bullet\dot\ta}A^{\tm\dot\tm})(\p^{\bullet}{}_{\dot\tb}A^{\tv\dot\tv})\nn\\
    &\equiv-\p_{\tm\dot\tm}A_{\ta(2)\tv\,\dot\tv}\p^{\ta}{}_{\dot\bullet}A^{\tm\dot\tm}\p^{\ta\dot\bullet}A^{\tv\dot\tv}-y^{\ta \dot\ta}y^{\ta\dot\tb}(\p_{\tm\dot\tm}A_{\ta(2)\tv\,\dot\tv})(\p_{\bullet\dot\ta}A^{\tm\dot\tm})(\p^{\bullet}{}_{\dot\tb}A^{\tv\dot\tv})\,\\
    V^{(1,2,1)}_{\boldsymbol{0}}&=-\frac{1}{2}\Big(\p_{\tm\dot\tm}A_{\tv\dot\tv}\p_{\tn\dot\bullet}A^{\tto\tn\tm\,\dot\tm}\p_{\tto}{}^{\dot\bullet}A^{\tv\dot\tv}+\p_{\tm\dot\tm}A_{\tv\dot\tv}\p_{\tn\dot\bullet}A^{\tn\tto\tm\,\dot\tm}\p_{\tto}{}^{\dot\bullet}A^{\tv\dot\tv}\Big)-y^{\ta\dot\ta}y^{\ta\dot\tb}(\p_{\tm\dot\tm}A_{\tv\dot\tv})(\p_{\bullet\dot\ta}A_{\ta(2)}{}^{\tm\dot\tm})(\p^{\bullet}{}_{\dot\tb}A^{\tv\dot\tv})\nn\,\\
    &\equiv-\p_{\tm\dot\tm}A_{\tv\dot\tv}\p_{\ta\dot\bullet}A^{\ta(2)\tm\,\dot\tm}\p_{\ta}{}^{\dot\bullet}A^{\tv\dot\tv}-y^{\ta\dot\ta}y^{\ta\dot\tb}(\p_{\tm\dot\tm}A_{\tv\dot\tv})(\p_{\bullet\dot\ta}A_{\ta(2)}{}^{\tm\dot\tm})(\p^{\bullet}{}_{\dot\tb}A^{\tv\dot\tv})\,.\\
    V^{(1,1,2)}_{\boldsymbol{0}}&=-\frac{1}{2}\Big(\p_{\tm\dot\tm}A_{\tv\dot\tv}\p_{\tn\dot\bullet}A^{\tm\,\dot\tm}\p_{\tto}{}^{\dot\bullet}A^{\tto\tn\tv\dot\tv}+\p_{\tm\dot\tm}A_{\tv\,\dot\tv}\p_{\tto\dot\bullet}A^{\tm\dot\tm}\p_{\tn}{}^{\dot\bullet}A^{\tto\tn\tv\dot\tv}\Big)-y^{\ta\dot\ta}y^{\ta\dot\tb}(\p_{\tm\dot\tm}A_{\tv\dot\tv})(\p_{\bullet\dot\ta}A^{\tm\dot\tm})(\p^{\bullet}{}_{\dot\tb}A_{\ta(2)}{}^{\tv\dot\tv})\nn\,\\
    &\equiv-\p_{\tm\dot\tm}A_{\tv\dot\tv}\p_{\ta\dot\bullet}A^{\tm\dot\tm}\p_{\ta}{}^{\dot\bullet}A^{\ta(2)\tv\,\dot\tv}-y^{\ta\dot\ta}y^{\ta\dot\tb}(\p_{\tm\dot\tm}A_{\tv\dot\tv})(\p_{\bullet\dot\ta}A^{\tm\dot\tm})(\p^{\bullet}{}_{\dot\tb}A_{\ta(2)}{}^{\tv\dot\tv})\,.
\end{align}
\end{subequations}
At $\Lambda=5$, we have 
\begin{align}
    \cV^5_{\boldsymbol{0}}=V^{(3,1,1)}+V^{(1,3,1)}+V^{(1,1,3)}+V^{(2,2,1)}+V^{(1,2,2)}+V^{(2,1,2)}\,.
\end{align}
Due to tracelessness of $A$'s, we get
\begin{align}
    V^{(3,1,1)}=V^{(1,3,1)}=V^{(1,1,3)}=0\,.
\end{align}
Next, we also have
\begin{subequations}
    \begin{align}
        V^{(2,2,1)}=V^{(1,2,2)}=V^{(2,1,2)}=0
    \end{align}
\end{subequations}
due to cancellation. This result is in agreement with the observation in \cite{Steinacker:2023cuf} that all cubic vertices vanish for odd $\Lambda$. Namely, one needs to have an even number of pairs $(\lambda,\bar\mu)$ so that the averaging can be non-trivial.

At $\Lambda=6$, due to tracelessness, $V^{(4,1,1)}=V^{(1,4,1)}=V^{(1,1,4)}=0$. However, we get
\begin{subequations}
    \begin{align}
        V^{(3,2,1)}_{\boldsymbol{0}}&=\frac{1}{3}\p_{\tm\dot\tm}A_{\ta(2)\tto\tn\tv\,\dot\tv}\p^{\tn}{}_{\dot\bullet}A^{\ta(2)\tm\,\dot\tm}\p^{\tto\dot\bullet}A^{\tv\,\dot\tv}+\cC_{(3,2,1)}(y)\,,\\
        V^{(3,1,2)}_{\boldsymbol{0}}&=\frac{1}{3}\p_{\tm\dot\tm}A_{\ta(2)\tto \tn\tv\,\dot\tv}\p^{\tn}{}_{\dot\bullet}A^{\tm\,\dot\tm}\p^{\tto\dot\bullet}A^{\ta(2)\tv\,\dot\tv}+\cC_{(3,1,2)}(y)\,,\\
        V^{(2,3,1)}_{\boldsymbol{0}}&=\frac{1}{3}\p_{\tm\dot\tm}A_{\ta(2)\tv\,\dot\tv}\p_{\tn \dot\bullet}A^{\ta(2)\tto\tn\tm\,\dot\tm}\p_{\tto}{}^{\dot\bullet}A^{\tv\,\dot\tv}+\cC_{(2,3,1)}(y)\,,\\
        V^{(1,3,2)}_{\boldsymbol{0}}&=\frac{1}{3}\p_{\tm\dot\tm}A_{\tv\,\dot\tv}\p_{\tn \dot\bullet}A^{\ta(2)\tto\tn\tm\,\dot\tm}\p_{\tto}{}^{\dot\bullet}A_{\ta(2)}{}^{\tv\,\dot\tv}+\cC_{(1,3,2)}(y)\,,\\
        V^{(1,2,3)}_{\boldsymbol{0}}&=\frac{1}{3}\p_{\tm\dot\tm}A_{\tv\dot\tv}\p_{\tn\dot\bullet}A_{\ta(2)}{}^{\tm\,\dot\tm}\p_{\tto}{}^{\dot\bullet}A^{\ta(2)\tto\tn\tv\,\dot\tv}+\cC_{(1,2,3)}(y)\,,\\
        V^{(2,1,3)}_{\boldsymbol{0}}&=\frac{1}{3}\p_{\tm\dot\tm}A_{\ta(2)\tv\,\dot\tv}\p_{\tn\dot\bullet}A^{\tm\,\dot\tm}\p_{\tto}{}^{\dot\bullet}A^{\ta(2)\tto\tn\tv\,\dot\tv}+\cC_{(2,1,3)}(y)\,;
    \end{align}
\end{subequations}
and 
\begin{align}
   24 V^{(2,2,2)}_{\boldsymbol{0}}=&+(\p_{\tm\dot\tm} A_{\tb\tb\tv\,\dot\tv})(\p_{\tn\dot\bullet}A^{\tto \tb\tm\,\dot\tm})(\p_{\tto}{}^{\dot\bullet}A^{\tb\tn\tv\,\dot\tv})+(\p_{\tm\dot\tm}A_{\tb(2)\tv\,\dot\tv})(\p_{\tn\dot\bullet}A^{\tb(2)\tm\,\dot\tm})(\p_{\tto}{}^{\dot\bullet}A^{\tto\tn\tv\,\dot\tv})\nn\\
    &+(\p_{\tm\dot\tm}A^{\tb\tn}{}_{\tv\,\dot\tv})(\p_{\tn\dot\bullet}A_{\tb(2)}{}^{\tm\,\dot\tm})(\p_{\tto}{}^{\dot\bullet}A^{\tto\tb\tv\,\dot\tv})+(\p_{\tm\dot\tm}A_{\tb(2)\tv\,\dot\tv})(\p_{\tto\dot\bullet}A_{\te}{}^{\tb\tm\,\dot\tm})(\p^{\tto\dot\bullet}A^{\tc\te\tv\,\dot\tv})\nn\\
    &+(\p_{\tm\dot\tm} A_{\tc(2)\tv\,\dot\tv})(\p_{\tn\dot\bullet}A^{\tto\tn\tm\,\dot\tm})(\p_{\tto}{}^{\dot\bullet}A^{\tc(2)\tv\,\dot\tv})+(\p_{\tm\dot\tm}A^{\tc\tn}{}_{\tv\,\dot\tv})(\p_{\tn\dot\bullet}A^{\tto\tc\tm\,\dot\tm})(\p_{\tto}{}^{\dot\bullet}A_{\tc(2)}{}^{\tv\,\dot\tv})\nn\\
    &-(\p_{\tm\dot\tm}A_{\tc\tb\tv\,\dot\tv})(\p_{\tto\dot\bullet}A^{\tb\td\tm\,\dot\tm})(\p^{\tto\dot\bullet}A_{\td}{}^{\tc\tv\,\dot\tv})+(\p_{\tm\dot\tm}A_{\tc(2)\tv\,\dot\tv})(\p_{\tn\dot\bullet}A^{\tc\tn\tm\,\dot\tm})(\p_{\tto}{}^{\dot\bullet}A^{\tto\tc\tv\,\dot\tv})\nn\\
    &+(\p_{\tm\dot\tm}A^{\tto\tb}{}_{\tv\,\dot\tv})(\p_{\tn\dot\bullet}A_{\tb(2)}{}^{\tm\,\dot\tm})(\p_{\tto}{}^{\dot\bullet}A^{\tb\tn\tv\,\dot\tv})+(\p_{\tm\dot\tm}A^{\tto\tc}{}_{\tv\,\dot\tv})(\p_{\tn\dot\bullet}A^{\tc\tn\tm\,\dot\tm})(\p_{\tto}{}^{\dot\bullet}A_{\tc(2)}{}^{\tv\,\dot\tv})\nn\\
    &+(\p_{\tm\dot\tm}A^{\tto\tn}{}_{\tv\,\dot\tv})(\p_{\tn\dot\bullet}A_{\tc(2)}{}^{\tm\,\dot\tm})(\p_{\tto}{}^{\dot\bullet}A^{\tc(2)\tv\,\dot\tv})+\cC_{(2,2,2)}(y)\,.
\end{align}
Here, $\cC_{(i,j,k)}(y)$ are contributions that break local Lorentz-invariance due to the presence of $y^{\alpha\dot\alpha}$. Observe that if the entries have little difference in spins, there will be more contractions with no simplification. Other higher-spin vertices $\cV_{\boldsymbol{0}}^{\Lambda}$  for $\Lambda\geq 7$ can be obtained analogously.  

$\bullet$ \underline{Extracting descendant vertices.} Let us now show how to extract descendant vertices at $\Lambda=4$ as a simple example. From \eqref{V0-Lambda-4}, we get
\begin{subequations}
    \begin{align}
        V_{\boldsymbol{1}}^{(\redone,1,1)}&=(\p_{\tm\dot\tm}\sA_{\tn\dot\tv})(\p^{\tn}{}_{\dot\bullet}A^{\tm\dot\tm})(\p_{\tv}{}^{\dot\bullet}A^{\tv\dot\tv})+(\p_{\tm\dot\tm}\sA_{\tto\dot\tv})(\p_{\tv\dot\bullet}A^{\tm\dot\tm})(\p^{\tto\dot\bullet}A^{\tv\dot\tv})\nn\\
        &\quad -(y_{\tv}{}^{\dot\ta}y^{\ta\dot\tb}+y^{\ta\dot\ta}y_{\tv}{}^{\dot\tb})(\p_{\tm\dot\tm}\sA_{\ta\dot\tv})(\p_{\bullet\dot\ta}A^{\tm\dot\tm})(\p^{\bullet}{}_{\dot\tb}A^{\tv\dot\tv})\,,\\
        V_{\boldsymbol{1}}^{(1,\redone,1)}&=(\p^{\tto}{}_{\dot\bullet}A_{\tv\dot\tv})(\p_{\tn\dot\bullet}\sA^{\tn\dot\tm})(\p_{\tto}{}^{\dot\bullet}A^{\tv\dot\tv})+(\p^{\tn}{}_{\dot\tm}A_{\tv\dot\tv})(\p_{\tn\dot\bullet}\sA^{\tto\dot\tm})(\p_{\tto}{}^{\dot\bullet}A^{\tv\dot\tv})\nn\\
        &\quad -(y^{\tm\dot\ta}y^{\ta\dot\tb}+y^{\ta\dot\ta}y^{\tm\dot\tb})(\p_{\tm\dot\tm}A_{\tv\dot\tv})(\p_{\bullet\dot\ta}\sA_{\ta}{}^{\dot\tm})(\p^{\bullet}{}_{\dot\tb}A^{\tv\dot\tv})\,,\\
        V_{\boldsymbol{1}}^{(1,1,\redone)}&=(\p_{\tm\dot\tm}A^{\tn}{}_{\dot\tv})(\p_{\tn\dot\bullet}A^{\tm\dot\tm})(\p_{\tto}{}^{\dot\bullet}\sA^{\tto\dot\tv})+(\p_{\tm\dot\tm}A^{\tto}{}_{\dot\tv})(\p_{\tn\dot\bullet}A^{\tm\dot\tm})(\p_{\tto}{}^{\dot\bullet}\sA^{\tn\dot\tv})\nn\\
        &\quad -(y^{\tv\dot\ta}y^{\ta\dot\tb}+y^{\ta\dot\ta}y^{\tv\dot\tb})(\p_{\tm\dot\tm}A_{\tv\dot\tv})(\p_{\bullet\dot\ta}A^{\tm\dot\tm})(\p^{\bullet}{}_{\dot\tb}\sA_{\ta}{}^{\dot\tv})\,,
    \end{align}
\end{subequations}
where the [red] index indicates the trace parts of $A$'s. 

Notice that $\Lambda=6$ is the first case where we can reach until the descendant vertex $\cV_{\boldsymbol{3}}^{3}=V_{\boldsymbol{3}}^{(\redone,\redone,\redone)}$. We leave this as an exercise for enthusiastic readers.

\paragraph{Vertices with diagonalized modes.} Since we have diagonalized $(A,\sA)$ into helicity-modes $\A^{\pm}$, it will be useful to substitute
\begin{align}
    A_{\ta(2s_1)\tm\,\dot\tm}=\frac{1}{2}\big(\A^+_{(\ta(2s)\tm)\,\dot\tm}+\A^-_{(\ta(2s)\tm)\,\dot\tm}\big)\,,\quad \sA_{\ta(2s-1)\,\dot\tm}=-\frac{\im}{2} \big(\A^+_{\ta(2s-1)\,\dot\tm}-\A^-_{\ta(2s-1)\,\dot\tm}\big)\,,
\end{align}
after obtaining the pre-diagonalized vertices. Note that one can have all configurations of ``helicities'' in the cubic vertices for Lorentzian \hsikkt, i.e.
\begin{align}
    S_3\ni \big\{V_{+++},V_{-++},V_{+-+},V_{++-},V_{--+},V_{-+-},V_{+--},V_{---}\big\}\,.
\end{align}
\subsection{Quartic action} 

Next, let us consider the quartic vertices which come from
\begin{align}
    S_4=\int \mho \{\sa^{\mu},\sa^{\nu}\}\{\sa_{\mu},\sa_{\nu}\}\approx -L_{\rm NC}^4e^{-3\tau}\int d^4y \int \tK\, \theta^{\tm\dot\tm\tn\dot\tn}\p_{\tm\dot\tm}\sa^{\ta\dot\ta}\p_{\tn\dot\tn}\sa^{\tb\dot\tb}\theta^{\tp \dot\tp \tq\dot\tq}\p_{\tp\dot\tp}\sa_{\ta\dot\ta}\p_{\tq\dot\tq}\sa_{\tb\dot\tb}\,.
\end{align}
This provides
\begin{align}
    S_4\approx -L_{\rm NC}^4e^{-3\tau}\int d^4y \sum_{s_i}\int_{\P^1}\tK\, (\ldots)\,,
\end{align}
where \small
\begin{align}
   (\ldots)&= \digamma_{\tI|\ta(2s_1),\tb(2s_2)}^{\tc(2s_3),\td(2s_4)|(\tm,\tn)|(\tp,\tq)}\p_{\tm\dot\bullet}\cA^{\ta(2s_1)\ta\,\dot\ta}\p_{\tn}{}^{\dot\bullet}\cA^{\tb(2s_2)\tb\,\dot\tb}\p_{\tp\dot\diamond}\cA_{\tc(2s_3)\ta\,\dot\ta}\p_{\tq}{}^{\dot\diamond}\cA_{\td(2s_4)\tb\,\dot\tb}\nn\\
   &+\digamma_{\tI\tI|\ta(2s_1),\tb(2s_2)|(\te_1,\te_2)}^{\tc(2s_3),\td(2s_4)|(\tm,\tn)}\p_{\tm\dot\bullet}\cA^{\ta(2s_1)\ta\,\dot\ta}\p_{\tn}{}^{\dot\bullet}\cA^{\tb(2s_2)\tb\,\dot\tb}(y^{\te_1\dot\tp}y^{\te_2\dot\tq}+y^{\te_2\dot\tp}y^{\te_1\dot\tq})\p_{\diamond\dot\tp}\cA_{\tc(2s_3)\ta\,\dot\ta}\p^{\diamond}{}_{\dot\tq}\cA_{\td(2s_4)\tb\,\dot\tb}\nn\\
   &+\digamma_{\tI\tI\tI|\ta(2s_1),\tb(2s_2)|(\tf_1,\tf_2)}^{\tc(2s_3),\td(2s_4)|(\tp,\tq)}(y^{\tf_1\dot\tm}y^{\tf_2\dot\tn}+y^{\tf_2\dot\tm}y^{\tf_1\dot\tn})\p_{\diamond\dot\tm}\cA^{\ta(2s_1)\ta\,\dot\ta}\p^{\diamond}{}_{\dot\tn}\cA^{\tb(2s_2)\tb\,\dot\tb}\p_{\tp\dot\bullet}\cA_{\tc(2s_3)\ta\,\dot\ta}\p_{\tq}{}^{\dot\bullet}\cA_{\td(2s_4)\tb\,\dot\tb}\nn\\
   &+ \digamma_{\tI\tV|\ta(2s_1),\tb(s_2)|(\tf_1,\tf_2)|(\te_1,\te_2)}^{\tc(2s_3),\td(2s_4)}(y^{\te_1\dot\tp}y^{\te_2\dot\tq}+y^{\te_2\dot\tp}y^{\te_1\dot\tq})(y^{\tf_1\dot\tm}y^{\tf_2\dot\tn}+y^{\tf_2\dot\tm}y^{\tf_1\dot\tn})\times\nn\\
   &\qquad \qquad\qquad \qquad\qquad\qquad \qquad \times \p_{\diamond\dot\tm}\cA^{\ta(2s_1)\ta\,\dot\ta}\p^{\diamond}{}_{\dot\tn}\cA^{\tb(2s_2)\tb\,\dot\tb}\p_{\circ\dot\tp}\cA_{\tc(2s_3)\ta\,\dot\ta}\p^{\circ}{}_{\dot\tq}\cA_{\td(2s_4)\tb\,\dot\tb}\,,
\end{align}
\normalsize
and the fiber factors in this case are
\begin{subequations}
    \begin{align}
    \digamma_{\tI|\ta(2s_1),\tb(2s_2)}^{\tc(2s_3),\td(2s_4)|(\tm,\tn)|(\tp,\tq)}=&\frac{1}{\langle\lambda\,\bar\mu\rangle^{s_1+s_2+s_3+s_4+2}}\lambda^{\tc(s_3)}\bar\mu^{\tc(s_3)}\lambda^{\td(s_4)}\bar\mu^{\td(s_4)}\times\nn\\
    &\times(\lambda^{\tm}\bar\mu^{\tn}+\lambda^{\tn}\bar\mu^{\tm})(\lambda^{\tp}\bar\mu^{\tq}+\lambda^{\tq}\bar\mu^{\tp})\lambda_{\ta(s_1)}\bar\mu_{\ta(s_1)}\lambda_{\tb(s_2)}\bar\mu_{\tb(s_2)}\,,\\
    \digamma_{\tI\tI|\ta(2s_1),\tb(s_2)|(\te_1,\te_2)}^{\tc(2s_3),\td(2s_4)|(\tm,\tn)}=&\frac{1}{\langle\lambda\,\bar\mu\rangle^{s_1+s_2+s_3+s_4+2}}\lambda^{\tc(s_3)}\bar\mu^{\tc(s_3)}\lambda^{\td(s_4)}\bar\mu^{\td(s_4)}\times\nn\\
    &\times(\lambda^{\tm}\bar\mu^{\tn}+\lambda^{\tn}\bar\mu^{\tm})\lambda_{\ta(s_1)}\bar\mu_{\ta(s_1)}\lambda_{\tb(s_2)}\bar\mu_{\tb(s_2)}\lambda_{\te_1}\bar\mu_{\te_2}\,,\\
    \digamma_{\tI\tI\tI|\ta(2s_1),\tb(s_2)|(\tf_1,\tf_2)}^{\tc(2s_3),\td(2s_4)|(\tp,\tq)}=&\frac{1}{\langle\lambda\,\bar\mu\rangle^{s_1+s_2+s_3+s_4+2}}\lambda^{\tc(s_3)}\bar\mu^{\tc(s_3)}\lambda^{\td(s_4)}\bar\mu^{\td(s_4)}\times\nn\\
    &\times(\lambda^{\tp}\bar\mu^{\tq}+\lambda^{\tq}\bar\mu^{\tp})\lambda_{\ta(s_1)}\bar\mu_{\ta(s_1)}\lambda_{\tb(s_2)}\bar\mu_{\tb(s_2)}\lambda_{\tf_1}\bar\mu_{\tf_2}\,,\\
    \digamma_{\tI\tV|\ta(2s_1),\tb(s_2)|(\tf_1,\tf_2)|(\te_1,\te_2)}^{\tc(2s_3),\td(2s_4)}=&\frac{1}{\langle\lambda\,\bar\mu\rangle^{s_1+s_2+s_3+s_4+2}}\lambda^{\tc(s_3)}\bar\mu^{\tc(s_3)}\lambda^{\td(s_4)}\bar\mu^{\td(s_4)}\times\nn\\
    &\times\lambda_{\ta(s_1)}\bar\mu_{\ta(s_1)}\lambda_{\tb(s_2)}\bar\mu_{\tb(s_2)}\lambda_{\tf_1}\bar\mu_{\tf_2}\lambda_{\te_1}\bar\mu_{\te_2}\,,
\end{align}
\end{subequations}
Similar to the cubic case, we denote \begin{align}
    \cV^{\Lambda}_{\boldsymbol{m}}=\sum_{s_i}V^{(s_1,s_2,s_3,s_4)}_{\boldsymbol{m}}\,,\qquad \boldsymbol{m}=0,1,2,3,4\,
\end{align}
to be the quartic vertex with maximal total spin, i.e. $\cV^{\Lambda}_{\boldsymbol{0}}$ and its associated descendants. Since there are 4 $\hs$-valued gauge fields entering the quartic vertices, we can have $\cV^{\Lambda-4}_{\boldsymbol{4}}$ if the maximal total spin is high enough. Below, we provide the simplest example of quartic vertices. 

$\bullet$ The lowest maximal total spin we can have is $\Lambda=4$. This corresponds to
\begin{align}
    V^{(1,1,1,1)}_{\boldsymbol{0}}=&+(\p_{\tm\dot\bullet}A^{\ta\dot\ta})(\p^{\tp \dot\bullet}A^{\tb\dot\tb})(\p_{\tp\dot\diamond}A_{\ta\dot\ta})(\p^{\tm\dot\diamond}A_{\tb\dot\tb})+(\p_{\tm\dot\bullet}A^{\ta\dot\ta})(\p^{\tp\dot\bullet}A^{\tb\dot\tb})(\p^{\tm}{}_{\dot\diamond}A_{\ta\dot\ta})(\p_{\tp}{}^{\dot\diamond}A_{\tb\dot\tb})\nn\\
    &+y_{\te_1}{}^{\dot\tp}y^{\te_2\dot\tq}\Big[\p_{\te_2\dot\bullet}A^{\ta\dot\ta}\p^{\te_1\dot\bullet}A^{\tb\dot\tb}\p_{\diamond\dot\tp}A_{\ta\dot\ta}\p^{\diamond}{}_{\dot\tq}A_{\tb\dot\tb}+\p^{\te_1}{}_{\dot\bullet}A^{\ta\dot\ta}\p_{\te_2}{}^{\dot\bullet}A^{\tb\dot\tb}\p_{\diamond\dot\tp}A_{\ta\dot\ta}\p^{\diamond}{}_{\dot\tq}A_{\tb\dot\tb}\Big]\nn\\
    &+y_{\tf_1}{}^{\dot\tm}y^{\tf_2\dot\tn}\Big[\p_{\diamond\dot\tm}A^{\ta\dot\ta}\p^{\diamond}{}_{\dot\tn}A^{\tb\dot\tb}\p_{\tf_2\dot\bullet}A_{\ta\dot\ta}\p^{\tf_1\dot\bullet}A_{\tb\dot\tb}+\p_{\diamond\dot\tm}A^{\ta\dot\ta}\p^{\diamond}{}_{\dot\tn}A^{\tb\dot\tb}\p^{\tf_1}{}_{\dot\bullet}A_{\ta\dot\ta}\p_{\tf_2}{}^{\dot\bullet}A_{\tb\dot\tb}\Big]\nn\\
    &+\frac{1}{2}\big(y_{\tf_1}{}^{\dot\tm}y^{\tf_2\dot\tn}+y^{\tf_2\dot\tm}y_{\tf_1}{}^{\dot\tn}\big)\big(y_{\tf_1}{}^{\dot\tp}y^{\tf_2\dot\tq}+y^{\tf_2\dot\tp}y_{\tf_1}{}^{\dot\tq}\big)\p_{\diamond\dot\tm}A^{\ta\dot\ta}\p^{\diamond}{}_{\dot\tn}A^{\tb\dot\tb}\p_{\circ\dot\tp}A_{\ta\dot\ta}\p^{\circ}{}_{\dot\tq}A_{\tb\dot\tb}\,.
\end{align}
Next, one can check that $\cV_{\boldsymbol{0}}^5=0$ at $\Lambda=5$ due to various cancellation. This is consistent with previous observation in \cite{Steinacker:2023cuf}. All other higher-spin cases are left as an exercise. 
\section{Some 3-point amplitudes}\label{app:B}
This appendix translates the 3-point amplitudes 
    \begin{align*}
        \cM_3(1^{\pm}_2,2^{\pm}_1,3^{\pm}_1)&\sim e^{-\frac{3}{2}\tau}L_{\rm NC}^2\,[2\,3]( \eps_2^{\pm,\tm\,\dot\tm}p_{1,\tm\dot\tm})\big(\eps^{\pm,\tv\dot\tv}_3\rho_2^{\ta}\rho_3^{\ta}\eps^{\pm}_{1,\ta(2)\tv\,\dot\tv}\big)\nn\\
        &\quad+ e^{-\frac{3}{2}\tau}L_{\rm NC}^2\langle 2\,3\rangle(\eps_2^{\pm,\tm\dot\tm}p_{1,\tm\dot\tm})(\eps_3^{\pm,\tv\dot\tv}y^{\ta\dot\ta}y^{\ta\dot\tb}\tilde \rho_{2\dot\ta}\tilde\rho_{3\dot\tb}\eps^{\pm}_{1,\ta(2)\tv\,\dot\tv})\nn\\
        &\quad+(\text{permutation})\,,\\
        \cM_3(1^{\pm}_1,2^{\pm}_2,3^{\pm}_1)&\sim e^{-\frac{3}{2}\tau}L_{\rm NC}^2\,[2\,3]\big(\eps_2^{\pm,\ta(2)\tm\,\dot\tm}\rho_{2\ta}\rho_{3\ta}p_{1,\tm\dot\tm}\big)(\eps_{3\,\pm}^{\tv\dot\tv}\eps^{\pm}_{1\tv\,\dot\tv})\nn\\
        &\quad + e^{-\frac{3}{2}\tau}L_{\rm NC}^2\langle 2\,3\rangle (p_1^{\tm\dot\tm}\eps^{\pm}_{2,\ta(2)\tm\,\dot\tm} y^{\ta\dot\ta}y^{\ta\dot\tb}\tilde\rho_{2\dot\ta}\tilde\rho_{3\dot\tb})(\eps_{1,\tv\dot\tv}^{\pm}\eps_{3,\pm}^{\tv\dot\tv})\nn\\
        &\quad+(\text{permutation})\,,\\
        \cM_3(1^{\pm}_1,2^{\pm}_1,3^{\pm}_2)&\sim e^{-\frac{3}{2}\tau}L_{\rm NC}^2\,[2\,3]( \eps_2^{\pm,\tm\,\dot\tm}p_{1,\tm\dot\tm})\big(\eps_{3\pm}^{\ta(2)\tv\,\dot\tv}\rho_{2\ta}\rho_{3\ta}\eps^{\pm}_{1,\tv\dot\tv}\big)\nn\\
        &\quad +e^{-\frac{3}{2}\tau}L_{\rm NC}^2\langle 2\,3\rangle(\eps_2^{\pm,\tm\dot\tm}p_{1,\tm\dot\tm})(\eps_{1,\pm}^{\tv\dot\tv}y^{\ta\dot\ta}y^{\ta\dot\tb}\eps^{\pm}_{3,\ta(2)\tv\,\dot\tv}\tilde\rho_{2\dot\ta}\tilde\rho_{3\dot\tb})\nn\\
        &\quad+(\text{permutation})\,
    \end{align*}
in Section \ref{sec:amplitudes} to spinor factorization form using the ansatz \eqref{YMhel-massless}. 

We begin with $\cM_3(1^{\pm}_2,2^{\pm}_1,3^{\pm}_1)$ where
\begin{subequations}
    \begin{align}
        \cM_3(1^{+}_2,2^{+}_1,3^{+}_1)&\sim 0\,,\\
        \cM_3(1^{-}_2,2^{+}_1,3^{+}_1)&\sim e^{-\frac{3}{2}\tau}L_{\rm NC}^2\,[2\,3]\frac{\langle \zeta\,1\rangle^2 [2\,1]\langle 2\,1\rangle\langle 3\,1\rangle[3\,\zeta]}{\langle 2\,\zeta\rangle \langle 3\,\zeta\rangle [1\,\zeta]}\nn\\
    &\quad +e^{-\frac{3}{2}\tau}L_{\rm NC}^2\langle 2\,3\rangle\frac{\langle \zeta\,1\rangle^2[2\,1][3\,\zeta]}{\langle 2\,\zeta\rangle\langle3\,\zeta\rangle[1\,\zeta]}\big(y^{\ta\dot\ta}y^{\ta\dot\ta}\rho^1_{\ta(2)}\tilde\rho_{2\dot\ta}\tilde\rho_{3\dot\tb}\big) -(2\leftrightarrow 3)\,,\\
        \cM_3(1^{+}_2,2^{-}_1,3^{+}_1)&\sim 0\,,\\
        \cM_3(1^{+}_2,2^{+}_1,3^{-}_1)&\sim -e^{-\frac{3}{2}\tau}L_{\rm NC}^2\,[2\,3]\frac{[2\,1][\zeta\,1]\langle 3\,\zeta\rangle^2}{[3\,\zeta]\langle1\,\zeta\rangle^2}\nn\\
    &\quad +e^{-\frac{3}{2}\tau}L_{\rm NC}^2\langle 2\,3\rangle\frac{\langle \zeta\,1\rangle[2\,1]\langle3\,\zeta\rangle[\zeta\,1]}{\langle 2\,\zeta\rangle[3\,\zeta]\langle 1\,\zeta\rangle^3}\big(y^{\ta\dot\ta}y^{\ta\dot\tb}\tilde\rho_{2\dot\ta}\tilde\rho_{3\dot\tb}\zeta_{\ta(2)}\big) +(1\leftrightarrow 2)\,,\\
        \cM_3(1^{-}_2,2^{-}_1,3^{+}_1)&\sim -e^{-\frac{3}{2}\tau}L_{\rm NC}^2\,[2\,3]\frac{\langle 2\,1\rangle^2\langle\zeta\,1\rangle[3\,\zeta]\langle3\,1\rangle}{\langle2\,\zeta\rangle\langle3\,\zeta\rangle}\nn\\
        &\quad +e^{-\frac{3}{2}\tau}L_{\rm NC}^2\langle2\,3\rangle\frac{\langle 2\,1\rangle\langle\zeta\,1\rangle[3\,\zeta]    }{[2\,\zeta]\langle 3\,\zeta\rangle}\big(y^{\ta\dot\ta}y^{\ta\dot\tb}\tilde \rho_{2\dot\ta}\tilde\rho_{3\dot\tb}\rho^1_{\ta(2)}\big)+(1\leftrightarrow 2)\,,\\
        \cM_3(1^{-}_2,2^{+}_1,3^{-}_1)&\sim 0\,,\\
        \cM_3(1^{+}_2,2^{-}_1,3^{-}_1)&\sim e^{-\frac{3}{2}\tau}L_{\rm NC}^2\,[2\,3]\frac{[\zeta\,1]\langle 2\,1\rangle[\zeta\,1]\langle 3\,\zeta\rangle^2\langle 2\,\zeta\rangle}{[2\,\zeta][3\,\zeta]\langle 1\,\zeta\rangle^3}\nn\\
        &\quad +e^{-\frac{3}{2}\tau}L_{\rm NC}^2\,\langle2\,3\rangle\frac{[\zeta\,1]^2\langle 2\,1\rangle\langle 3\,\zeta\rangle}{[2\,\zeta][3\,\zeta]\langle 1\,\zeta\rangle^3}\big(y^{\ta\dot\ta}y^{\ta\dot\tb}\tilde\rho_{2\dot\ta}\tilde\rho_{3\dot\tb}\zeta_{\ta(2)}\big)-(2\leftrightarrow 3) \,,\\
        \cM_3(1^{-}_2,2^{-}_1,3^{-}_1)&\sim0 \,.
    \end{align}
\end{subequations}
We observe that in the massless sector where we can impose the space-like gauge conditions \eqref{space-like-condition}, the 3-point amplitudes simplify. In particular, they depend only on spinors and not on the explicit Lorentz-violation structures, i.e. the structures with $y$'s factors. In what follows, we will suppress $y$-dependent contributions to simplify our expressions.

Next, let us consider $\cM_3(1^{\pm}_1,2^{\pm}_2,3^{\pm}_1)$ where
\begin{subequations}
    \begin{align}
        \cM_3(1^{-}_1,2^{+}_2,3^{+}_1)&\sim e^{-\frac{3}{2}\tau}L_{\rm NC}^2\Big(\,[2\,3]\frac{\langle \zeta\,1\rangle^2[3\,\zeta][2\,1]}{[1\,\zeta]\langle2\,\zeta\rangle^2}+\cC_{(1^-_1,2^+_2,3^+_1)}(y)\Big)-(2\leftrightarrow 3)\,,\\
        \cM_3(1^{+}_1,2^{+}_2,3^{-}_1)&\sim e^{-\frac{3}{2}\tau}L_{\rm NC}^2\Big(\,\,[2\,3]\frac{\langle 3\,\zeta\rangle^2[\zeta\,1][2\,1]}{[3\,\zeta]\langle 2\,\zeta\rangle^2}+\cC_{(1^+_1,2^+_2,3^-_1)}(y)\Big)+(1\leftrightarrow 2)\,,\\
        \cM_3(1^{-}_1,2^{-}_2,3^{+}_1)&\sim 0\,,\\
        \cM_3(1^{+}_1,2^{-}_2,3^{-}_1)&\sim 0 \,,
    \end{align}
\end{subequations}
and zero otherwise. Recall that $\cC(y)$ are contributions that are $y$-dependent. Finally,
\begin{subequations}
    \begin{align}
        \cM_3(1^{-}_1,2^{+}_1,3^{+}_2)&\sim e^{-\frac{3}{2}\tau}L_{\rm NC}^2\Big(\,[2\,3]\frac{\langle \zeta\,1\rangle^2[2\,1] [3\,\zeta]}{\langle3\,\zeta\rangle^2[1\,\zeta]}+\cC_{(1^-_1,2^+_1,3^+_2)}(y)\Big)-(2\leftrightarrow 3)\,,\\
        \cM_3(1^{+}_1,2^{+}_1,3^{-}_2)&\sim 0\,,\\
        \cM_3(1^{-}_1,2^{-}_1,3^{+}_2)&\sim e^{-\frac{3}{2}\tau}L_{\rm NC}^2\Big(\,[2\,3]\frac{\langle 2\,1\rangle\langle\zeta\,2\rangle\langle\zeta\,1\rangle[3\,\zeta]}{[2\,\zeta]\langle3\,\zeta\rangle^2}+\cC_{(1^-_1,2^-_1,3^+_2)}(y)\Big)+(1\leftrightarrow 2)\,,\\
        \cM_3(1^{+}_1,2^{-}_1,3^{-}_2)&\sim 0\,.
    \end{align}
\end{subequations}
Besides these amplitudes, we also wish to consider the simplest example where spin-3 arises. Let us consider for example
\begin{align*}
    V_{\boldsymbol{0}}^{3,2,1}=\frac{1}{3}\p_{\tm\dot\tm}A_{\ta(2)\tto\tn\tv\,\dot\tv}\p^{\tn}{}_{\dot\bullet}A^{\ta(2)\tm\,\dot\tm}\p^{\tto\dot\bullet}A^{\tv\,\dot\tv}+\cC_{(3,2,1)}(y)
\end{align*}
and write it in a helicity basis as
\begin{align}
     V_{\boldsymbol{0}}^{(3,2,1)}=\frac{1}{24}(\p_{\tm\dot\tm}\A^{\pm}_{\ta(2)\tto\tn\tv\,\dot\tv})(\p^{\tn}{}_{\dot\bullet}\A_{\pm}^{\ta(2)\tm\,\dot\tm})(\p^{\tto\dot\bullet}\A_{\pm}^{\tv\,\dot\tv})+\cC_{(3,2,1)}(y)\,.
\end{align}
Its associated 3-point amplitudes read
\begin{align}
    \cM_3(1_3^{\pm},2_2^{\pm},1_1^{\pm})\sim e^{-\frac{3}{2}\tau}&L_{\rm NC}^2\Big([2\,3]\big(\rho_2^{\tn}\rho_3^{\tto}\eps_{3\pm}^{\tv\,\dot\tv}\eps^{\pm}_{1,\ta(2)\tto\tn\tv\,\dot\tv}\eps_{2\pm}^{\ta(2)\tm\,\dot\tm}p_{1\tm\,\dot\tm}\big)+\cC_{(3,2,1)}(y)\Big)\nn\\
    &\qquad \qquad \qquad \qquad+\text{(permutations)}\,.
\end{align}
Using the polarization ansatz \eqref{YMhel}, we can write the above amplitudes as
\begin{subequations}
    \begin{align}
          \cM_3(1^{+}_3,2^{+}_2,3^{-}_1)&\sim 0\,,\\
        \cM_3(1^{+}_3,2^{-}_2,3^{-}_1)&\sim e^{-\frac{3}{2}\tau}L_{\rm NC}^2\Big(\,[2\,3]\frac{\langle 2\,1\rangle[\zeta\,1]\langle 2\,\zeta\rangle^2\langle3\,\zeta\rangle^2[\zeta\,1]}{\langle1\,\zeta\rangle^5[2\,\zeta][3\,\zeta]}+\cC_{(1^+_3,2^-_2,3^-_1)}(y)\Big)-(2\leftrightarrow 3)\,,\\
        \cM_3(1^{-}_3,2^{+}_2,3^{+}_1)&\sim e^{-\frac{3}{2}\tau}L_{\rm NC}^2\Big(\,[2\,3]\frac{\langle \zeta\,1\rangle^3[2\,1]\langle 2\,1\rangle\langle 3\,1\rangle[3\,\zeta]\langle \zeta\,1\rangle}{[1\,\zeta]\langle 2\,\zeta\rangle^3\langle3\,\zeta\rangle}+\cC_{(1^+_3,2^+_2,3^+_1)}(y)\Big)-(2\leftrightarrow 3)\,,\\
        \cM_3(1^{-}_3,2^{-}_2,3^{+}_1)&\sim e^{-\frac{3}{2}\tau}L_{\rm NC}^2\Big(\,[2\,3]\frac{\langle 2\,1\rangle^4\langle3\,1\rangle[3\,\zeta]\langle \zeta\,1\rangle}{[2\,\zeta]\langle 3\,\zeta\rangle} +\cC_{(1^-_3,2^-_2,3^+_1)}(y)\Big)+(1\leftrightarrow 2)\,.
    \end{align}
\end{subequations}
In the massless sector, all explicit Lorentz-violation contributions of the 3-point amplitudes vanish using \eqref{eq:space-like-condition-spinors}. The above reduces to
\begin{subequations}
    \begin{align}
        \cM_3(1^{+}_3,2^{-}_2,3^{-}_1)&\sim 0\,,\\
        \cM_3(1^{-}_3,2^{+}_2,3^{+}_1)&\sim 0\,,\\
        \cM_3(1^{-}_3,2^{-}_2,3^{+}_1)&\sim 0\,.
    \end{align}
\end{subequations}



\setstretch{0.8}
\footnotesize
\bibliography{twistor}
\bibliographystyle{JHEP}

\end{document}